\documentclass[runningheads]{llncs}
\usepackage[T1]{fontenc}
\usepackage{graphicx}
\usepackage{mathpartir}
\usepackage{quiver}
\usepackage{enumitem}
\usepackage{wrapfig}
\usepackage{fancyvrb}
\usepackage{amssymb, amsmath, stmaryrd}
\usepackage{ wasysym }
\usepackage{todonotes}
\usepackage[all]{xy}
\usepackage{mathtools}
\usepackage{cite}
\usepackage{xspace}

\usepackage{hyperref}
\usepackage{cleveref}[capitalize]
\usepackage{pedrohaa}
\DeclareFontFamily{U}{dmjhira}{}
\DeclareFontShape{U}{dmjhira}{m}{n}{ <-> dmjhira }{}

\newcommand{\Set}{\cat{Set}}

\newcommand{\gradedRelDef}[1]{\mathrel{\mathcal{A}}_{#1}}

\newcommand{\op}{\mathrm{op}}

\newcommand{\clat}{\cat{CLat}_\sqcap}
\newcommand{\liftF}{\mathcal{H}}

\newcommand{\rew}{\text{rew}}
\newcommand{\true}{\mathrm{True}}
\newcommand{\false}{\mathrm{False}}
\newcommand{\dar}{DAwR\xspace}
\newcommand{\doctrine}{adjoint gluing doctrine\xspace}
\newcommand{\Doctrine}{Adjoint gluing doctrine\xspace}

\newcommand{\totalgraded}[2]{\Delta^*#1^{#2}} 
\newcommand{\ladjDelta}[1][-]{\exists g.#1}

\newcommand{\ol}{\overline}
\newcommand{\subdist}{{\mathcal{D}_{\leq 1}}}

\newcommand{\fibgraded}[2]{\Delta^*#1^{#2}} 
\newcommand{\Coalg}[1]{\mathbf{Coalg}(#1)}
\newcommand{\id}{\mathrm{id}}
\newcommand{\pullbackmark}[2]{\save ;p+<.8pc,0pc>:(0,-1)::%
(#1) *{\phantom{Z}} %
;p+(#2)-(0,0) **@{-}%
;p-(#1)+(0,0) *{\phantom{Z}} **@{-} \restore}
\newcommand{\twocat}[1]{\cat{#1}}
\newcommand{\rloop}[2][-]{\save \POS!R(.7) \ar@(ru,rd)^#1{#2} \restore}
\newcommand{\lloop}[2][-]{\save \POS!L(.7) \ar@(lu,ld)_#1{#2} \restore}
\newcommand{\uloop}[2][-]{\save \POS!U(.7) \ar@(lu,ru)^(.8){#2} \restore}
\newcommand{\dloop}[2][-]{\save \POS!D(.7) \ar@(rd,ld)^(.8){#2} \restore}

\crefname{section}{\S\!}{Sects.}
\Crefname{section}{\S\!}{Sects.}
\crefname{proposition}{Prop.}{Props.}
\Crefname{proposition}{Prop.}{Props.}
\crefname{definition}{Def.}{Defs.}
\Crefname{definition}{Def.}{Defs.}
\crefname{appendix}{App.}{Appx.}
\Crefname{appendix}{App.}{Appx.}
\crefname{theorem}{Thm.}{Thms.}
\Crefname{theorem}{Thm.}{Thms.}
\crefname{lemma}{Lem.}{Lems.}
\Crefname{lemma}{Lem.}{Lems.}
\crefname{corollary}{Cor.}{Cors.}
\Crefname{corollary}{Cor.}{Cors.}
\crefname{figure}{Fig.}{Figs.}
\Crefname{figure}{Fig.}{Figs.}
\crefname{example}{Example}{Examples}
\Crefname{example}{Example}{Examples}

\usepackage{mdframed}
\usepackage{xcolor}
\usepackage{tcolorbox}

\definecolor{cornellred}{RGB}{196,18,48}
\definecolor{umichblue}{HTML}{00274C}

\begin{document}
\title{A Framework for Coalgebraic Reward-Sensitive Bisimulation}
%
%
\author
{Pedro H. Azevedo de Amorim\inst{1} \and
Mayuko Kori\inst{2}\and
Koko Muroya\inst{3}}
\authorrunning{Azevedo de Amorim et al.}
%
\institute{University of Bath, Bath, UK \and
Kyoto University, Research Institute for Mathematical Sciences, Kyoto, Japan  \and
Ochanomizu University, Tokyo, Japan}
\maketitle              
\begin{abstract}
  In this paper we present a framework for modelling 
  \emph{reward-sensitive bisimulations},
  that is, bisimulations that account for quantitative differences such as accumulated rewards.
  To capture both qualitative and quantitative aspects uniformly,
  we consider two interacting notions of bisimulation:
  a graded variant that tracks bounded reward differences,
  and an ungraded one that abstracts from them.

  Our characterization
  of these notions 
  is done in the fibrational and
  coalgebraic approach to (bi)simulation initiated by Hermida and Jacobs.
  To formally relate the graded and ungraded notions,
  we deploy categorical gluing, a standard technique in
  categorical logic. Furthermore, we show that this construction
  interacts well with standard coalgebra concepts, such as final
  coalgebras, and that it yields  a unified characterization in terms of  combined notions of bisimulations under mild assumptions.

  In order to demonstrate the versatility of our approach, 
  we show how
  it encompasses 
  various bisimulation notions for different kinds of systems, including 
  relation-based bisimulations for automata with rewards  and
  metric-based notions of bisimulations for labelled Markov processes.
\end{abstract}
\section{Introduction}

Quantitative extensions of (bi)simulations have been extensively studied 
for systems equipped with quantitative data. 
For example, over the past few years there has been much
work on behavioural metrics and relaxations thereof to reason about distances
of states in transition systems\cite{DBLP:conf/csl/DesharnaisS26,DBLP:conf/stacs/BeoharG0MFSW24,DBLP:journals/logcom/SprungerKDH21}. 
Similarly, there has been work on defining
notions of transitions systems where transitions may incur cost \cite{DBLP:journals/tcs/Weber94,DBLP:journals/iandc/Ibaraki76,DBLP:journals/iandc/Ibaraki78b,DBLP:conf/lics/AlurDDRY13,schutzenberger1961definition} and
bisimulation for such transition systems \cite{DBLP:conf/forte/KiehnA05}.

In this paper, we focus on reward-sensitive (bi)simulations,
which capture not only whether two states are just similar, but also similar up to a relative reward.
For example, imagine a variant of finite automata where each transition
has a reward associated with it. In such a scenario, one might be
interested in understanding whether one state can accept the same
string as a state similar to it which receives more reward. 
While traditional notions of bisimulation relations
are effective at reasoning about qualitative properties of transition
systems, such as whether two states accept the same language, they
cannot cope with quantitative properties such as distances and relative cost.
This motivates 
a graded refinement of the usual notion of bisimulation,
where the grade indexes the tolerance difference in accumulated reward.
Such graded refinement has appeared in the literature as
\emph{amortised bisimulations}~\cite{DBLP:conf/forte/KiehnA05}
for labelled transition systems with costs. Yet, it is still uncertain
if there is a common abstraction that encompasses these reward-sensitive
notions of bisimulation with more metric-based examples. Furthermore,
there is no categorical/coalgebraic theory that explains when such
graded bisimulations can be seen as a refinement of an ungraded
notion of bisimulation.

Motivated by these kinds of applications, 
we explore coalgebraic approaches to graded refinements
of bisimulation relations and its interactions with ungraded
bisimulations. 
In particular,
we propose a coalgebraic and
fibred axiomatization of graded and ungraded notions of bisimulation,
obtained by applying
categorical gluing---an important construction from
categorical logic---to the coalgebra world. We use this
construction to define the ``coherent'' pairing of two different notions
of coalgebra along a well-behaved morphism between them. 
We show that it
provides
a purely formal way to relate the coalgebraic structure of graded and ungraded bisimulations,
and to identify
minimal structural assumptions for
ensuring that an ungraded notion of bisimulation can be seen as the ``collapse''
of a graded notion of bisimulation, i.e. that the graded bisimulation is a
refinement of the ungraded one.

In order to substantiate the canonicity of our construction, we prove
an abstract soundness theorem that relates the graded and ungraded
notions of coalgebra.  We also show that, in a way we make precise in
\Cref{th:commacoalg}, these ``glued coalgebras'' are the universal way
of combining the notions of graded and ungraded bisimulations. Then,
prove \Cref{th:gluedequiv}, which means that such coalgebras form a
section-retraction pair of graded and ungraded coalgebras. We conclude
by characterizing the final coalgebra in the glued model.

Our work builds on the line of work initiated by Hermida and
Jacobs~\cite{DBLP:journals/iandc/HermidaJ98}, who showed that fibrations
provide rich structures for defining abstract notions of
bisimulation as lifting functors up fibrations.
As a general framework, we adopt a standard setting
based on fibrations and fibred adjunctions,
which we call an \emph{\doctrine{}}. 
This axiomatization
provide a common language for explaining traditional notions of bisimulations 
purely coalgebraically and generalizing them beyond binary relations to quantitative structures, such as 
pseudometric.
We illustrate that our framework instantiates to
a range of examples,
including
simulations on deterministic automata with rewards (\dar), which is a new notion as far as the authors know,
and existing ones such as
amortised bisimulations and approximate bisimulations~\cite{DBLP:conf/qest/DesharnaisLT08,DBLP:conf/csl/DesharnaisS26}.

\paragraph*{Related work}
As discussed above,
there has been extensive research on quantitative refinements of classical bisimulation relations,
and many coalgebraic frameworks have been proposed to capture such refinements.
For probabilistic systems,
several coalgebraic and fibrational accounts of behavioural metrics have been developed~\cite{DBLP:conf/fsttcs/BaldanBKK14,DBLP:journals/logcom/SprungerKDH21}.
Another line of work concerns 
\emph{Graded semantics}~\cite{DBLP:conf/calco/MiliusPS15,DBLP:conf/concur/DorschMS19},
which has been proposed to organise behavioural equivalences along the linear-time and branching-time spectrum~\cite{DBLP:conf/concur/Glabbeek90}.
Although our notion of graded bisimulation shares the idea of grading,
the usage is different:
while grades in graded semantics usually represent the depth of transitions,
grades in our graded bisimulation represent tolerances on the differences of rewards.
In this paper, we focus on reward-sensitive bisimulations, 
which differ from the above quantitative refinements.
A closely related development is the recent coalgebraic framework for \emph{amortised analysis}
by Grodin and Harper~\cite{grodin2024amortized}.
Their work mainly focuses on generalizing \emph{potential functions},
rather than on bisimulations. A point of convergence between their work
and the notions of bisimulation considered here is their use of 2-dimensional
data to bound the cost between two different systems. They encoded this as
categories enriched in partial orders while we adopt an even finer grained
approach by making the relative cost be a part of the data.

\paragraph*{Organization}
We start by studying a concrete example based on deterministic automata
with rewards, 
showing how it naturally gives rise to a graded notion of simulations
that
bound
the relative reward between states. 
We then formulate \dar and their simulations coalgebraically.

Motivated by the categorical structures that emerge from this coalgebraic treatment, we
introduce an \doctrine{}, as a general framework capturing the interactions between graded and ungraded (bi)simulations.
Once we set this up, we show how the two settings can be glued together, 
and
we establish a general soundness theorem 
relating the graded (bi)simulation with
its ungraded variant
within the glued setting.

Finally, we illustrate the generality of our framework by showing that it encompasses several useful examples.

\newcommand{\Fda}{F_\text{DA}}
\newcommand{\Fhatda}{\hat{F}_\text{DA}}
\section{Motivating example: simulations on deterministic automata with rewards} \label{sec:prelim}
We start with a motivating example, namely simulations
for deterministic automata and their reward-sensitive extensions.

  A \emph{deterministic automaton (or DA)} is a tuple $(S, \Sigma, \delta, \alpha)$ consisting of
a set of states $S$, an input alphabet $\Sigma$, a transition function $\delta\colon S \to S^\Sigma$,
 and a set of accepting states $\alpha\subseteq S$.
 It is well known that DAs admit a notion of simulations that characterises language inclusion.

\begin{definition} \label{def:ungradedsim}
  An \emph{(ungraded) simulation for a DA $(S, \Sigma, \delta, \alpha)$} is a binary relation $R \subseteq S \times S$ such that for every $(s, s') \in R$: $s \in \alpha$ implies $s' \in \alpha$ and
for every $a \in \Sigma$, $(\delta(s, a), \delta(s', a)) \in R$.   The \emph{similarity} is defined as the largest simulation.  
\end{definition}
\begin{proposition} \label{prop:da_sound}
  The similarity is the set of pairs such that every word accepted from the first is also accepted from the second.
  Hence,
    if $R$ is a simulation on a DA,
  then every $(s, s') \in R$ satisfies this language inclusion property.
\end{proposition}

We then move on to a reward-sensitive setting.
A \emph{deterministic automaton with rewards (or \dar)} is a deterministic automaton equipped with a reward function $\rew\colon S \to \nat^\Sigma$.
The acceptance of words is defined as in DAs. 
For a state $s \in S$ and a word $w \in \Sigma^*$, the \emph{accumulated reward} $\rew(s, w)$ is defined as the natural extension of the reward function by
$\rew(s, \epsilon) = 0$ and $\rew(s, a w) = \rew(s, a) + \rew(\delta(s, a), w)$ for each $a \in \Sigma$ and $w \in \Sigma^*$.
A \dar with an initial state $s$ thus induces a partial function $\rew(s, -)\colon \Sigma^* \rightharpoonup \nat$ that assigns to each accepted word its accumulated reward and is undefined for words outside the language.
We adopt this partial function as the semantic basis for defining simulations between \dar{}s.

These kinds of automata has been studied in the classical literature~\cite{DBLP:journals/iandc/Ibaraki76,DBLP:journals/tcs/Weber94}.
It is closely related to \emph{deterministic weighted automata} over $(\nat, +, 0)$.
A key difference lies in the presence of the acceptance condition:
a \dar assigns rewards only to accepted words, while a weighted automaton assigns rewards to all words regardless of acceptance.

For \dar{}s, we extend the classical notion of simulation so as to 
 capture quantitative bounds on relative rewards
between similar states.
In this setting, we do not require that rewards are preserved exactly. 
Instead, for similar states $s$ and $s'$,
we require that every word $w$ accepted from $s$ is also accepted from $s'$,
and 
moreover,

the difference in accumulated rewards never exceeds a bound $n$ along any partial run.
This motivates the introduction of a grading by natural numbers, where the index $n$ indicates the maximal tolerated reward difference.
\begin{definition} \label{def:gradedsim}
  A \emph{graded simulation on a \dar $(S, \Sigma, \delta, \rew, \alpha)$} is a family of binary relations $\gradedRelDef{} = \{\gradedRelDef{n}\}_{n \in \nat}$ over $S$ such that if $(s, s') \in \gradedRelDef n$:
  \begin{itemize}
  \item $s \in \alpha$ implies $s' \in \alpha$,
  \item for every $a \in \Sigma$, $\rew(s, a) \leq n + \rew(s', a)$,
  \item 
    for every $a \in \Sigma$, $(\delta(s, a), \delta(s', a)) \in \gradedRelDef {n - \rew(s, a) + \rew(s', a)}$.
  \end{itemize}
  The \emph{graded similarity} is defined as the largest simulation.
\end{definition}

It is easy to see that this notion of graded simulation is sound with respect to language inclusion with rewards:
\begin{proposition} \label{prop:dar_sound}
  The graded similarity of a \dar is the graded binary relation $\{\gradedRelDef{n}\}_{n \in \nat}$ collecting all pairs $(s, s')$ such that 
  $(\star)$: every word $w$ accepted from $s$ is also accepted from $s'$ 
  and $\rew(s, w') \leq n + \rew(s', w')$ for each prefix $w'$ of $w$.
  
  Hence, 
  if $\gradedRelDef{}$ is a graded simulation on a \dar,
  then
  $(s, s') \in \gradedRelDef n$ implies that it satisfies the condition $(\star)$.
\end{proposition}

A graded simulation on a \dar induces an ungraded simulation on the underlying DA by taking the union of all relations in the family.
\begin{proposition} \label{prop:dasim_trans}
  A graded simulation $\{\gradedRelDef{n}\}_n$ on a \dar
  induces a simulation $\bigcup_n \gradedRelDef{n}$ on the underlying DA.
\end{proposition}
In \Cref{sec:gluing}, we use comma objects to give a more conceptual framing of this property.

Observing that the indices of a graded simulation behave like tolerances, it is natural to impose the following additivity conditions on a graded simulation:
\begin{enumerate}
  \item For each $s \in S$, $(s, s) \in \gradedRelDef{0}$,
  \item For each $n, m \in \nat$, 
  $\{(x, y) \mid \exists z.~(x, z) \in \gradedRelDef{n} \text{ and }(z, y) \in \gradedRelDef{m}\} \subseteq \gradedRelDef{n+m}$.
\end{enumerate}
The first condition reflects that the reward difference between a state and itself is always $0$,
and the second condition reflects 
that if the reward difference from $x$ to $z$ is at most $n$, and from $z$ to $y$ is at most $m$, then the reward difference from $x$ to $y$ is at most $n + m$.
A graded simulation satisfying these condition is called \emph{lax monoidal}.
In fact, from a categorical perspective, these conditions mean that $\gradedRelDef{}$ is a lax monoidal functor
from the monoidal category $(\nat, 0, +)$ to the monoidal category of endorelations with 
 relation compositions as the monoidal product.
 We will introduce fibrational frameworks for these ungraded and (lax monoidal) graded simulations in \Cref{sec:coalg_for_graded}.

\section{Background: coalgebraic framework for simulations}
We briefly recall a coalgebraic account of (bi)simulations, which generalizes the classical definitions for transition systems.
In the coalgebraic framework, a transition system is modelled as a coalgebra $c\colon S \to FS$ of an endofunctor $F\colon \cat{C} \to \cat{C}$.
To capture relational reasoning, one often employs a fibration $p\colon \cat{E} \to \cat{C}$ whose total category $\cat{E}$ is a category of relations over objects in $\cat{C}$.
This perspective follows the established framework of fibrational approaches to coalgebraic (bi)simulations, 
see e.g.~\cite{DBLP:journals/iandc/HermidaJ98,DBLP:journals/logcom/SprungerKDH21}.
For details of fibrations in general,
we refer the reader to Chapters 1 and 9 of \cite{DBLP:books/daglib/0023251}.
In this paper, we focus in particular on a class of 
posetal fibrations known as $\clat$-fibrations, which are a special case of topological functors~\cite{HERRLICH1974125}.

\begin{definition}[$\clat$-fibration]
  A \emph{$\clat$-fibration} is a fibration $p \colon \cat{E} \to \cat{C}$
  such that each fibre $\cat{E}_X$ is a complete lattice and for each morphism
 $f \colon X \to Y$ in $\cat{C}$,  the reindexing functor $f^* \colon \cat{E}_Y \to \cat{E}_X$
  preserves all meets.
\end{definition}
We use $\sqsubseteq$ for the order in each fibre,
and $\leq$ for standard orders, such as the usual order on natural numbers
or the order on booleans $2 \coloneqq \{\false, \true\}$ satisfying $\false \leq \true$ and $\true \not \leq \false$.
It is known that a $\clat$-fibration is faithful so that 
for each $P, Q \in \cat{E}$ and $f\colon pP \to pQ$, there is at most one morphism $P \to Q$ in $\cat{E}$ above $f$.
\newcommand{\dotarrow}{\mathrel{\overset{\scalebox{0.4}{\ $\bullet$}}{\rightarrow}}}
Therefore, we often write $f\colon P \dotarrow Q$ for the existence of a morphism $P \to Q$ above $f\colon pP \to pQ$, which is equivalent to $P \sqsubseteq f^*Q$ in $\cat{E}_{pP}$.
\begin{definition}
  Let $p\colon \cat{E} \to \cat{C}$ and $q\colon \cat{F} \to \cat{D}$ be fibrations.
  Let $F\colon \cat{C} \to \cat{D}$ and $\hat{F}\colon \cat{E} \to \cat{F}$ be functors. We say that $\hat{F}$ is a \emph{lifting of $F$ along $p$ and $q$} if the equality $q \circ \hat{F} = F \circ p$ holds.
  If moreover $\cat{E} = \cat{F}$ and $p = q$,
  we simply say that $\hat{F}$ is a \emph{lifting of $F$ along $p$}.
  A functor $\hat{F}\colon \cat{E} \to \cat{F}$ is said to be \emph{fibred} if 
  it is a lifting of some $F\colon \cat{C} \to \cat{D}$ and it preserves cartesian morphisms.
\end{definition}

Given a $\clat$-fibration $p\colon \cat{E} \to \cat{C}$, we fix a lifting $\hat{F}\colon \cat{E} \to \cat{E}$ of $F\colon \cat{C} \to \cat{C}$ along $p$.
A coalgebra $c\colon S \to FS$ in $\cat{C}$ models a transition system,
and
a coalgebra $c\colon R \dotarrow \hat{F}R$ (or such $R$)
is called a \emph{simulation} on the coalgebra $c\colon S \to FS$.
In general,
the lifting $\hat{F}$ is often instantiated by specific constructions suited for simulations, such as relation liftings or codensity liftings~\cite{DBLP:journals/jlp/KurzV16,DBLP:journals/logcom/SprungerKDH21}.
However, since the present paper does not rely on particular structures of liftings,
we place no such restriction and work with arbitrary liftings.

Since the fibre $\cat{E}_S$ is a complete lattice and the endofunction $c^*\hat{F}$ on this fibre is monotone,
the Knaster--Tarski theorem~\cite{Tarski1955ALF} guarantees
the existence of the greatest fixed point $\nu (c^*\hat{F})$.
The fixed point represents the largest simulation on $c$,
which is called the \emph{similarity}.
We often refer to simulations and similarities as \emph{bisimulations} and \emph{bisimilarities}
when they are symmetric in an appropriate sense
depending on the context.

\begin{example}[(Ungraded) simulations on DAs] \label{eg:da}
  Let
  $(S, \Sigma, \delta, \alpha)$ be a deterministic automaton (DA).
  It can be represented as a coalgebra $c \coloneqq \langle \alpha, \delta \rangle \colon S \to 2 \times S^\Sigma$ of the functor $\Fda \coloneqq 2 \times (-)^\Sigma\colon \cat{Set} \to \cat{Set}$.

  Consider the category $\cat{Rel}$ 
  whose objects are endorelations $(X, R \subseteq X \times X)$ and morphisms are relation preserving functions, i.e.~$f\colon (X, R) \to (Y, Q)$ are 
  functions $f\colon X \to Y$ satisfying $(x_1, x_2) \in R$ implies $(fx_1, fx_2) \in Q$.
  We often simply write $R$ for $(X, R)$.
  The forgetful functor $p\colon \cat{Rel} \to \cat{Set}$ is
  a $\clat$-fibration: the fibre $\cat{Rel}_X$ over a set $X$ is the complete lattice of binary relations over $R$ ordered by inclusion,
  and for a function $f\colon X \to Y$, the reindexing functor $f^*$ is given by $f^*R \coloneqq \{(x_1, x_2) \mid (fx_1, fx_2) \in R\}$.

We define a lifting $\Fhatda\colon \cat{Rel} \to \cat{Rel}$ of $\Fda$ by
\begin{displaymath}
  \Fhatda(R) \coloneqq \left\{\big((b_1, \tau_1), (b_2, \tau_2)\big)  \mid
    b_1 \leq b_2 \text{ and }
    \forall a \in \Sigma.~\big(\tau_1(a), \tau_2(a)\big) \in R
  \right\}.
\end{displaymath}
With this setup,
 an ungraded simulation $R$ on a DA $c$ (see \Cref{def:ungradedsim}) is precisely an object $R \in \cat{Rel}$
 such that
 $c\colon R \dotarrow \Fhatda R$.
\end{example}

We recall the following definition ({cf.~\cite{DBLP:journals/iandc/HermidaJ98}}),
which enables us to consider a category of simulations 
and that of transition systems
as categories $\Coalg{\hat{F}}$ and $\Coalg{F}$ of coalgebras, respectively, connected by the functor $\Coalg{p}$ induced by the fibration $p$.
\begin{definition} \label{def:coalg}
  For an endofunctor $F\colon \cat{C} \to \cat{C}$, we write $\Coalg{F}$ for the category of $F$-coalgebras and $F$-coalgebra homomorphisms.
  Assuming that $F\colon \cat{C} \to \cat{C}$, $G\colon \cat{D} \to \cat{D}$, $H\colon \cat{C} \to \cat{D}$ are functors, and $\alpha\colon HF \Rightarrow GH$ is a natural transformation, there is a functor $\Coalg{H}_\alpha\colon \Coalg{F} \to \Coalg{G}$ defined by
  \begin{align*}
  \Coalg{H}_\alpha(c) &\coloneqq \alpha_X \circ Hc \text{ for each $c\colon X \to FX$ and} \\
  \Coalg{H}_\alpha(f) &\coloneqq Hf \text{ for each $f\colon c \to d$ in $\Coalg{F}$}.
  \end{align*}
  We sometimes write $\Coalg{H}$ for $\Coalg{H}_{\id}$ when $HF = GH$.
\end{definition}

\section{Coalgebraic framework for graded simulations} \label{sec:coalg_for_graded}

In this section we provide a canonical way of coalgebraically formulating
the graded notion of simulation presented in the previous section. We do this
by refining an ``ungraded'' notion of coalgebraic simulation. In the broader
context of this paper, the goal of this section is to explore the categorical
structures that naturally emerge when making ungraded and graded notions of
relations interact. Then, in the next section, we use such structures as a
starting point for defining more general notions of graded and ungraded simulations.

We start by stating some definitions and results about fibred limits. These can
very naturally capture the kinds of relations we are interested in. 
Then, we provide
a coalgebraic proof of \Cref{prop:dar_sound}.

\subsection{Graded simulations}

We begin by defining some notions from the
fibred category theory literature.

\begin{definition}[{cf.~\cite[Def.~3.3.4]{DBLP:books/daglib/0072949}}] \label{def:gradedfib}
  Let $I$ be a small category, $\cat{C}$ be a category with $I$-colimits,
  and $p\colon \cat{E} \to \cat{C}$ be a $\clat$-fibration.
  Then we define the \emph{$I$-graded fibration} $\fibgraded{p}{I}\colon \totalgraded{\cat{E}}{I} \to \cat{C}$ of $p$ by the pullback diagram:
  \begin{displaymath}
    \xymatrix@R=2em{
      \totalgraded{\cat{E}}{I} \pullbackmark{1, 0}{0, 1}\ar[d]_-{\fibgraded{p}{I}} \ar[r] &\cat{E}^I \ar[d]^{p^I} \\
      \cat{C} \ar[r]^-\Delta &\cat{C}^I,
    }
  \end{displaymath}
  where $\Delta$ is the diagonal functor.
  Since $p^I = p \circ (-)$ is a $\clat$-fibration, $\fibgraded{p}{I}$ is also a $\clat$-fibration~\cite[Prop.~2.9]{DBLP:journals/ngc/KomoridaKHKHEH22}.
  We overload our notation and write $\Delta\colon \cat{E} \to \totalgraded{\cat{E}}{I}$ for 
  the mediating map induced by $p\colon \cat{E} \to \cat{C}$ and $\Delta\colon \cat{E} \to \cat{E}^I$.
\end{definition}
The functor category $\cat{E}^I$ should be thought of category of $I$-graded relations, which is fibred over the category of ``graded sets'' $\cat{C}^I$. 
Because we do not want to work with $I$-indexed sets, 
we take the pullback above along the diagonal functor $\Delta$.
The $\clat$-fibration $p$ has fibred $I$-colimits~\cite[Ex.~9.2.4]{DBLP:books/daglib/0023251}. 
This implies that $\Delta$ is a fibred functor of type $p \to \fibgraded{p}{I}$ and it has a fibred left adjoint $\ladjDelta$.

\begin{equation} \label{eq:fibadj}
\begin{tikzcd}
	{\cat{E}} && {\totalgraded{\cat{E}}{I}} \\
	& {\cat{C}}
	\arrow[""{name=0, anchor=center, inner sep=0}, "\Delta", curve={height=-6pt}, from=1-1, to=1-3]
	\arrow["p"', from=1-1, to=2-2]
	\arrow[""{name=1, anchor=center, inner sep=0}, "{\ladjDelta}", curve={height=-6pt}, from=1-3, to=1-1]
	\arrow["\fibgraded{p}{I}", from=1-3, to=2-2]
	\arrow["\dashv"{anchor=center, rotate=90}, draw=none, from=1, to=0]
\end{tikzcd}
\end{equation}

In \Cref{eg:da}, we showed that an ungraded simulation on a deterministic automaton $c\coloneqq \langle \alpha, \delta \rangle \colon S \to \Fda S$ is an object $R$ such that $c\colon R \dotarrow \Fhatda(R)$ in $\cat{Rel}$ above $\langle \alpha, \delta \rangle$.
We now extend this setting to graded simulations on \dar{}s, using the fibration introduced in \Cref{def:gradedfib}.
A \dar $(S, \Sigma, \delta, \alpha, \rew)$ can then be expressed as a coalgebra $c\coloneqq \langle \alpha, \langle \delta, \rew \rangle \rangle\colon S \to \Fda (S \times \nat)$.
\begin{example}[Graded simulations on \dar{}s] \label{eg:dar}
  Let $\nat$ be the set of natural numbers, regarded as a discrete category,
  and consider the fibration $p\colon \cat{Rel} \to \cat{Set}$.
  The category $\totalgraded{\cat{Rel}}{\nat}$ has 
  as objects
  $\nat$-graded relations $(X, \{\gradedRelDef{n}\}_{n \in \nat})$ 
  and 
  as morphisms
  $f\colon (X, \gradedRelDef{}) \to (Y, \gradedRelDef{}')$ 
  the
  functions $f\colon X \to Y$ preserving relations for each grade.
  The fibration $\fibgraded{p}{\nat}\colon \totalgraded{\cat{Rel}}{\nat} \to \cat{Set}$ is the forgetful functor,
  and 
  the functor $\ladjDelta$ maps a graded relation $\gradedRelDef{}$ to $\bigcup_n \gradedRelDef{n}$.

  To capture graded simulations, we define a lifting of $\Fda(- \times \nat)$ on the category $\totalgraded{\cat{Rel}}{\nat}$.
  We construct this lifting as a composition of two components:
  one handling the deterministic automaton structure and the other accounting for the reward part.
  Specifically,
  we define
  two liftings
  $\Fhatda^\nat\colon \totalgraded{\cat{Rel}}{\nat} \to \totalgraded{\cat{Rel}}{\nat}$ of $\Fda$ 
  and 
  $M\colon \totalgraded{\cat{Rel}}{\nat} \to \totalgraded{\cat{Rel}}{\nat}$ of $(-) \times \nat$, 
  corresponding respectively to 
  the deterministic automaton component and the reward component, by
  \begin{align*}
    \Fhatda^\nat(\mathcal{A})_n &\coloneqq
    \left\{
      ((b_1, \tau_1), (b_2, \tau_2))  \mid
      b_1 \leq b_2,
      \text{ and }
      \forall a \in \Sigma.~(\tau_1(a), \tau_2(a)) \in \gradedRelDef{n}
    \right\}, \\
    M(\gradedRelDef{})_n &\coloneqq \{((s_1, n_1), (s_2, n_2)) \mid n_1 \leq n +  n_2\text{ and } (s_1, s_2) \in \gradedRelDef{(n-n_1 + n_2)}\}.
  \end{align*}
  Then a graded simulation on a \dar $\langle \alpha, \langle \delta, \rew \rangle \rangle\colon S \to \Fda(S \times \nat)$ is an object $\gradedRelDef{} \in \totalgraded{\cat{Rel}}{\nat}$ such that $\langle \alpha, \langle \delta, \rew \rangle \rangle \colon \gradedRelDef{} \dotarrow \Fhatda^\nat M(\gradedRelDef{})$.
\end{example}

Based on the setup of \Cref{eg:da} and \Cref{eg:dar},
\Cref{prop:dasim_trans} can be reformulated in coalgebraic terms as follows.
\begin{proposition}[Soundness, cf. \cref{ap:proof_soundness_sim}]
  \label{prop:soundness}
  There is a natural transformation $\epsilon\colon (\ladjDelta) \circ M \Rightarrow \ladjDelta$ above $\pi_1\colon (-) \times \nat \Rightarrow \id$ 
  and a vertical natural transformation $\alpha\colon (\ladjDelta) \circ \Fhatda^\nat \Rightarrow \Fhatda \circ (\ladjDelta)$,
  such that
  the following commutes:
    \begin{displaymath}
    \xymatrix@C=7em{
      \Coalg{\Fhatda^\nat\circ M} \ar[r]^-{\Coalg{\ladjDelta}_{\gamma}} \ar[d]_{\Coalg{\fibgraded{p}{\nat}}} &\Coalg{\Fhatda} \ar[d]^{\Coalg{p}} \\
      \Coalg{\Fda \circ ((-) \times \nat)} \ar[r]_-{\Coalg{\id}_{\Fda\pi_1}} &\Coalg{\Fda}.
    }
    \end{displaymath}
    We define $\gamma\coloneqq (\Fhatda * \epsilon) \circ (\alpha * M) \colon (\ladjDelta) \Fhatda^\nat  M \Rightarrow \Fhatda  (\ladjDelta)$
    where
    $\circ$ and $*$ are vertical and horizontal compositions of natural transformations, respectively.
\end{proposition}

\begin{corollary}
  For a \dar $\langle \alpha, \langle \delta, \rew \rangle \rangle\colon S \to \Fda(S \times \nat)$ and its underlying DA $\langle \alpha, \delta \rangle\colon S \to \Fda (S)$,
  we have
  $(\ladjDelta)\nu(\langle \alpha, \langle \delta, \rew \rangle \rangle^*\Fhatda^\nat M)  \sqsubseteq \nu(\langle \alpha, \delta \rangle^*\Fhatda)$.
\end{corollary}
In other words, the union of the graded similarity of a \dar is included in the similarity of the underlying DA (cf.~\Cref{prop:da_sound} and \Cref{prop:dar_sound}).

\subsection{Graded simulations in a lax monoidal setting}

As we have argued in \Cref{sec:prelim}, ``basic'' graded relations do not
capture an important basic structure of reward-sensitive phenomena: they satisfy
graded transitivity properties, i.e. the reward structure should be additive.
We achieve this by using lax monoidal fibrations with strong monoidal fibres \cite{zawadowski}.

\begin{definition}
  We say that a $\clat$-fibration $p\colon \cat{E} \to \cat{C}$ is \emph{lax monoidal} if
  each fibre category $\cat{E}_X$ of the $\clat$-fibration $p$ is a monoid $(\cat{E}_X, \star_X, 1_X)$
  such that 
  $\star_X$ is order-preserving in each argument and
  $f^*P \star_X f^*Q \sqsubseteq f^*(P \star_Y Q)$, and $1_X \sqsubseteq f^*1_Y$ for each $f\colon X \to Y$ in $\cat{C}$ and $P, Q \in \cat{E}_Y$.
\end{definition}

\begin{definition} \label{def:lax_monoidal_ver}
  In the setting of \Cref{def:gradedfib},
  we further assume that $I$ is a monoidal category $(I, \oplus, 0)$ and 
  $p$ is lax monoidal.
  Then we define the category 
  $\ol{\totalgraded{\cat{E}}{I}}$ 
  as the full subcategory of $\totalgraded{\cat{E}}{I}$ consisting of $\gradedRelDef{}\colon I \to \cat{E}_X$ such that $\gradedRelDef{i} \star \gradedRelDef{j} \sqsubseteq \gradedRelDef{i \oplus j}$ and $1 \sqsubseteq \gradedRelDef{0}$ for every $i, j \in I$.
\end{definition}

While this is restriction alludes to some connections to the triangle inequality and
metric reasoning,
we lose the nice fibred adjunction we had between the categories of graded and ungraded
relations. Luckily, because the fibrations are faithful
and the adjunction is fibred over the identity, we can show that that
the right adjoint $\Delta \colon \cat{E} \to \totalgraded{\cat{E}}{I}$ is
full and faithful and $\ladjDelta \circ \Delta = \id$. This lets us
lift the adjunction up the inclusion $\ol{\totalgraded{\cat{E}}{I}} \hookrightarrow \totalgraded{\cat{E}}{I}$ by using the following lemma, which follows from a direct
calculation.
\begin{lemma}
  Consider the pullback below in $\clat(\cat{C})$.
\[\begin{tikzcd}
	{\cat{E}'} && {\cat{F}'} \\
	{\cat{E}} && {\cat{F}} \\
	& {\cat{C}}
	\arrow["{\Delta'}", hook, from=1-1, to=1-3]
	\arrow[hook, from=1-1, to=2-1]
	\arrow["\lrcorner"{anchor=center, pos=0.125}, draw=none, from=1-1, to=2-3]
	\arrow["{ ~}"{description}, from=1-1, to=3-2]
	\arrow[hook, from=1-3, to=2-3]
	\arrow["{~}"{description}, from=1-3, to=3-2]
	\arrow["\Delta", hook, from=2-1, to=2-3]
	\arrow[from=2-1, to=3-2]
	\arrow[from=2-3, to=3-2]
\end{tikzcd}\]
If the fibred inclusion $\Delta \colon \cat{E} \hookrightarrow \cat{F}$ has a left adjoint $L$
such that $L \circ \Delta = \id$ and $L$ restricts along the inclusion $\cat{F}' \hookrightarrow \cat{F}$ then $L$ lifts to a fibred left adjoint to the functor $\Delta' \colon \cat{E}' \hookrightarrow \cat{F}'$ such that $L' \circ \Delta' = \id$.
\end{lemma}
By instantiating this lemma with our lax monoidal restriction above and denoting $\ol{\cat{E}}$ to be the pullback of $\Delta$ along the inclusion $\ol{\totalgraded{\cat{E}}{I}} \hookrightarrow \totalgraded{\cat{E}}{I}$, we obtain a lax monoidal restriction of the diagram \eqref{eq:fibadj}.
\begin{proposition}
  If the functor $\ladjDelta\colon \totalgraded{\cat{E}}{I} \to \cat{E}$ restricts to $\ol{\totalgraded{\cat{E}}{I}} \to \ol{\cat{E}}$, then we have the following diagram of $\clat$-fibrations and a fibred adjunction:
\begin{minipage}{\linewidth}
  \centering
\begin{tikzcd}
	{\ol{\cat{E}}} && {\ol{\totalgraded{\cat{E}}{I}}} \\
	& {\cat{C}}
	\arrow[""{name=0, anchor=center, inner sep=0}, hook, "\Delta", curve={height=-6pt}, from=1-1, to=1-3]
	\arrow["\ol{p}"', from=1-1, to=2-2]
	\arrow[""{name=1, anchor=center, inner sep=0}, "\ladjDelta", curve={height=-6pt}, from=1-3, to=1-1]
	\arrow["\ol{\fibgraded{p}{I}}", from=1-3, to=2-2]
	\arrow["\dashv"{anchor=center, rotate=90}, draw=none, from=1, to=0]
\end{tikzcd}\par
\hfill
\end{minipage}
\end{proposition}
Explicitly,
meets  in $\ol{\totalgraded{\cat{E}}{I}}$ are given by those in $\totalgraded{\cat{E}}{I}$. Furthermore, the left adjoint $\ol{\Delta^* \cat{E}^I} \to \Delta^* \cat{E}^I$ maps $\gradedRelDef{} \in \totalgraded{\cat{E}}{I}$ to
$\ol{\gradedRelDef{}} = \bigwedge \{\mathcal{B} \in \ol{\totalgraded{\cat{E}}{I}} \mid \gradedRelDef{} \sqsubseteq \mathcal{B}\}$.
We call the fibration $\ol{\fibgraded{p}{I}}\colon \ol{\totalgraded{\cat{E}}{I}} \to \cat{C}$ the \emph{lax monoidal $I$-graded fibration} of $p$.

While these definitions may seem quite abstract, their concrete instantiations can
be quite natural, as we now show.

\begin{example}[Simulations on deterministic automata] \label{eg:dar_mon}
  We revisit \Cref{eg:dar} in a lax monoidal setting.
  The $\clat$-fibration $p\colon \cat{Rel} \to \cat{Set}$ is lax monoidal
  with monoids $(\cat{Rel}_X, \star_X, \mathrm{Refl}_X)$ defined by 
  \begin{align*}
   & R_1 \star_X R_2 \coloneqq \{(x, y) \mid \exists z.~(x, z) \in R_1 \text{ and }(z, y) \in R_2\}\\
    &\mathrm{Refl}_X \coloneqq \{(x, x) \mid x \in X\}.
  \end{align*}
  Consider the category $\omega$ of natural numbers with a unique morphism $m \to n$ if and only if $m \leq n$.
  It forms a monoidal category $(\omega, +, 0)$.
  Then $\ol{\totalgraded{\cat{Rel}}{\omega}}$ is the category 
  of $\omega$-graded relations $\gradedRelDef{}\colon \omega \to \cat{Rel}_X$ satisfying $\mathrm{Refl} \subseteq \gradedRelDef{0}$ and $\gradedRelDef{n} \star_X \gradedRelDef{m} \subseteq \gradedRelDef{n+m}$ and morphisms are functions between the underlying sets that preserve relations at each grade.
  A direct calculation shows that $\ol{\cat{Rel}}$ is equivalent to the category of reflexive  and transitive binary relations. By restricting the liftings $\Fhatda$ and $\Fhatda^\nat$ defined in \Cref{eg:da,eg:dar} to those on $\ol{\cat{Rel}}$ and $\ol{\totalgraded{\cat{Rel}}{\omega}}$, respectively, we obtain the lax monoidal variants of ungraded and graded simulations on DAs and \dar{}s, as discussed at the end  of \Cref{sec:prelim}.
\end{example}

\section{Refinements Through Glued Fibrations} \label{sec:gluing}

The previous section shows that there is plenty of
categorical structure present in such graded refinements
of notions of simulation. In this section we abstract
away from the DAwR example and show that we can prove
a generalization of \Cref{prop:soundness} by using categorical
gluing, a technique coming from categorical logic
and programming language theory.

We begin by defining the \doctrine{}, a fibred adjunction that
generalizes \Cref{eg:dar}, providing an abstract framework for
representing graded/ungraded notions of simulation. An important
insight from our work is that neither the ungraded nor the graded
notions of simulation are the ``ground truth''. Instead, they should
be thought of as different analyses that serve different purposes.

In order to unify these analyses we combine both fibrations into a single
one by using a fibred generalization of the comma category, often used in
proofs by logical relations in categorical logic and programming language theory.

This comma construction allows for a fairly conceptual and simple proof of a
generalization of \Cref{prop:soundness}. Finally, in order to substantiate the
claim that this glued model is a natural setting to do coalgebra in, we
study some of its structural properties such as final coalgebras and comma object
preservation.
\subsection{A Doctrine for Refinements of Simulations}

Something key in the structure present in the \dar example is the
existence of adjoint functors that allow you to move between the graded
and ungraded worlds. Therefore, this is the basis of our axiomatization,
which we now define.
Such axiomatization is common in categorical logic and coalgebraic bisimulations~\cite{DBLP:books/daglib/0072949,DBLP:journals/logcom/SprungerKDH21}.
Although our paper focuses on ungraded/graded simulations, the framework and the results in this section can be applied to other settings as well.

\begin{definition}
  An \emph{\doctrine{}} 
  is a tuple $(F, p, p^g, \hat{F},  \hat{F}^g, \ladjDelta, \Delta, \alpha)$ 
  where $F\colon \cat{C} \to \cat{C}$ is an endofunctor,
  $p\colon \cat{E} \to \cat{C}$ and $p^g\colon \cat{E}^g \to \cat{C}$ are $\clat$-fibrations,
  $\hat{F}\colon \cat{E} \to \cat{E}$ and $\hat{F}^g\colon \cat{E}^g \to \cat{E}^g$ are liftings of $F$ along $p$ and $p^g$, respectively,
  $\ladjDelta \dashv \Delta$ is a fibred adjunction between $p$ and $p^g$,
  and $\alpha$ is a natural transformation $ (\ladjDelta) \circ \hat{F}^g \Rightarrow \hat{F} \circ (\ladjDelta)$.

\begin{equation} \label{eq:doctorine}
\begin{tikzcd}
	{\cat{E}} && {\cat{E}^g} \\
	& {\cat{C}}
	\arrow["{\hat{F}}", from=1-1, to=1-1, loop, in=150, out=210, distance=5mm]
	\arrow[""{name=0, anchor=center, inner sep=0}, "\Delta", curve={height=-6pt}, from=1-1, to=1-3]
	\arrow["p"', from=1-1, to=2-2]
	\arrow[""{name=1, anchor=center, inner sep=0}, "{\ladjDelta}", curve={height=-6pt}, from=1-3, to=1-1]
	\arrow["{\hat{F}^g}"', from=1-3, to=1-3, loop, in=30, out=330, distance=5mm]
	\arrow["p^g", from=1-3, to=2-2]
	\arrow["{F}"', from=2-2, to=2-2, loop, in=300, out=240, distance=5mm]
	\arrow["\dashv"{anchor=center, rotate=90}, draw=none, from=1, to=0]
\end{tikzcd}
\end{equation}
\end{definition}

Intuitively, $\cat{E}$ is a
category of ungraded relations, while $\cat{E}^g$ is the category of
graded relations.  The functor $\Delta$ embeds ungraded relations into
graded relations by defining the ``constant'' graded relation, while
$\ladjDelta$ maps graded relations to ungraded relations by taking the
union over all grades.

Then for a coalgebra $c\colon S \to F(S)$ in $\cat{C}$ representing transition systems with rewards, 
we have two simulation notions:

\begin{itemize}
  \item A coalgebra $R \to \hat{F}(R)$ in $\cat{E}$ above $c$, which is called an \emph{ungraded simulation}.
  \item A coalgebra $\gradedRelDef{} \to \hat{F}^g(\gradedRelDef{})$ in $\cat{E}^g$ above $c$, which is called a \emph{graded simulation}.
\end{itemize}
Intuitively, ungraded simulations are standard simulations without reward information, while graded simulations can provide quantitative bounds on relative rewards between two similar states.
The greatest fixed points  $\nu(c^*\hat{F})$ and $\nu(c^*\hat{F}^g)$ are called the \emph{ungraded} and \emph{graded similarities}, respectively.

\begin{example}[Doctrine for \dar]
  The data introduced in \Cref{eg:dar_mon} forms an \doctrine{}.
  Concretely,
  letting $p$ be the fibration $\cat{Rel} \to \Set$,
  the tuple $(\Fda, \ol{p}, \ol{\fibgraded{p}{\omega}}, \Fhatda, \Fhatda^g, \ladjDelta, \Delta, \alpha)$
  is an \doctrine{},
  where $\alpha$ is the canonical natural transformation given by $\Fhatda = \ladjDelta\circ \Fhatda^g \circ \Delta$.
  Note that this is slightly different from \Cref{eg:da}, where the transitions are not costful.
\end{example}

\subsection{A Glued Fibration}

While the \doctrine{} abstraction is relatively simple, it requires us
to keep track of more than one notion of coalgebra. We now show that it is
possible to bring ideas from categorical logic to combine these two notions
of coalgebra into a single one. First, we define a generalization of
the concept of comma categories.
\begin{definition}[Comma objects {\cite[Section 2.1]{weber2007yoneda}}]
  A comma object of a cospan  $f \colon a \to c$ and $g \colon b \to c$ in a
  2-category $\mathcal{K}$ is a lax universal square
\[\begin{tikzcd}
	x & b \\
	a & c
	\arrow["q", from=1-1, to=1-2]
	\arrow["p"', from=1-1, to=2-1]
	\arrow["g", from=1-2, to=2-2]
	\arrow["\alpha", between={0.3}{0.7}, Rightarrow, from=2-1, to=1-2]
	\arrow["f"', from=2-1, to=2-2]
\end{tikzcd}\]
satisfying the following universal properties:
\begin{description}
\item[One-dimensional:] Every 2-cell $\alpha' \colon f p' \to g q'$ factors uniquely through a morphism $h \colon x' \to x$ as $\alpha h$.
\item[Two-dimensional:] For every pair of 1-cells $h_1, h_2 \colon y \to x$.
If there are 2-cells $\beta_1 \colon h_1p \to h_2p$ and $\beta_2 \colon h_1q \to h_2 q$
such that an exchange law holds:
\[\begin{tikzcd}
	&&& y & x \\
	y & x & b & x & b \\
	x & a & c & a & c
	\arrow["{h_2}", from=1-4, to=1-5]
	\arrow["{h_1}"', from=1-4, to=2-4]
	\arrow["q", from=1-5, to=2-5]
	\arrow["{h_2}", from=2-1, to=2-2]
	\arrow["{h_1}"', from=2-1, to=3-1]
	\arrow["q", from=2-2, to=2-3]
	\arrow["p"', from=2-2, to=3-2]
	\arrow[""{name=0, anchor=center, inner sep=0}, "g", from=2-3, to=3-3]
	\arrow["{\beta_2}", between={0.3}{0.8}, Rightarrow, from=2-4, to=1-5]
	\arrow["q", from=2-4, to=2-5]
	\arrow[""{name=1, anchor=center, inner sep=0}, "p"', from=2-4, to=3-4]
	\arrow["g"{description}, from=2-5, to=3-5]
	\arrow["{\beta_1}", between={0.3}{1}, Rightarrow, from=3-1, to=2-2]
	\arrow["p"', from=3-1, to=3-2]
	\arrow["\alpha"', between={0.2}{0.8}, Rightarrow, from=3-2, to=2-3]
	\arrow["f"', from=3-2, to=3-3]
	\arrow["\alpha", between={0.3}{0.8}, Rightarrow, from=3-4, to=2-5]
	\arrow["f"', from=3-4, to=3-5]
	\arrow["{=}"{description}, draw=none, from=0, to=1]
\end{tikzcd}\]
then there is a unique 2-cell $\beta \colon h_1 \to h_2$ such that
$p\beta = \beta_1$ and $q\beta = \beta_2$.
\end{description}
\end{definition}

In programming language and type theory, comma objects are often
used as a way of combining two distinct models of a theory along a
morphism. Such models can then be used to construct ``logical
relations'' models of your language. We draw inspiration from these
lines of work and show how these \doctrine{} can be encapsulated as
coalgebras over a comma category. Concretely, we glue
$\cat{E}^g$ and $\cat{E}$ along the functor $\ladjDelta$,
resulting in the comma category:
\[\begin{tikzcd}
	{\cat{X}} & {\cat{E}} \\
	{\cat{E}^g} & {\cat{E}}
	\arrow[from=1-1, to=1-2]
	\arrow[from=1-1, to=2-1]
	\arrow["\id", from=1-2, to=2-2]
	\arrow["\varphi", between={0.3}{0.7}, Rightarrow, from=2-1, to=1-2]
	\arrow["{\ladjDelta}"', from=2-1, to=2-2]
\end{tikzcd}\]

Concretely, the objects of $\cat{X}$ are triples $(\gradedRelDef{},
R, f \colon \ladjDelta[\gradedRelDef{g}] \to R)$ and its
morphisms are pairs of morphisms $h_1 \colon \gradedRelDef{} \to
\mathcal{B}$ and $h_2 \colon R \to Q$ making the
appropriate diagram commute. In intuitive terms, the morphism 
$\ladjDelta[\gradedRelDef{g}] \to R$ is a witness to the fact that the
union of the graded relation is a subset of the ungraded relation.
Furthermore, note that $\cat{X}$ comes equipped with two projections
$\cat{X} \to \cat{E}^g$ and $\cat{X} \to \cat{E}$, satisfying the
universal property of comma objects.
When instantiated to the $\dar$ example,
a glued notion of simulation takes the following form:
a glued simulation over a state space $X$ is a tuple
$(\gradedRelDef{}, R, f)$ of a graded $\dar$ simulation
$\gradedRelDef{}$, a $\dar$ ungraded simulation over $X$, and a relation-preserving morphism $f\colon \bigcup_g \gradedRelDef{g} \to R$.
In particular, when 
$f=\id$, this means that the ungraded simulation $\bigcup_g \gradedRelDef{g}$ is contained in $R$.

The power of the comma category becomes clear when we note that we
can use its universal property to readily combine the liftings
$\hat{F}$ and $\hat{F}^g$ into a single one---this is where we use
the natural transformation
$(\ladjDelta) \circ \hat{F^g} \Rightarrow \hat{F} \circ (\ladjDelta)$.

While the category $\cat{X}$ does capture aspects
of both categories $\cat{E}$ and $\cat{E}^g$,
we are interested in understanding how it relates to
the base category $\cat{C}$. We do this by using
a slight generalization of a result by Hermida~\cite[Proposition 4.14]{hermida1999some}.
\begin{lemma}[{cf. \Cref{ap:proof_comma_fib}}] \label{lem:comma_fib}
  The 2-subcategory $\clat(\cat{C}) \hookrightarrow \cat{Fib}(\cat{C})$
  has comma objects.
\end{lemma}
Once we have this set-up, we want to show that we are indeed capturing
some kind of refinement at the simulation level. For instance, we
want to prove that if we have a graded simulation $\gradedRelDef{g}$,
its union $\ladjDelta[\gradedRelDef{g}]$ is an ungraded and reward-free
simulation. Furthermore, there is a universal way ``pushing forward''
the graded relation. Once again,
this result follows from  properties of comma objects.
\begin{lemma}[{cf. \Cref{app:adjointcomma}}]
  \label{lem:adjointcomma}
  For every 1-cell $f \colon x \to y $, the projection 1-cells
  $(f \downarrow \id) \to x$ has a left adjoint.
\end{lemma}
We now show that in the setting of \doctrine{}, we can universally glue both
functor liftings together. First, we need a definition.
\begin{definition}
  Let $\twocat{K}$ be a 2-category. Given a 0-cell $C \colon \twocat{K}$ and a
  1-cell $f \colon C \to C$, we define the 2-category $\cat{Lift}(f)$
  as follows
  \begin{description}
  \item[0-cells] Triples $(E, p \colon E \to C, \hat{f} \colon E \to E)$ such
    that $p \hat{f} = f p$.
  \item[1-cells] Pairs $(h, \lambda) \colon (E, p, \hat{f}) \to (E', p', \hat{f'})$, where $h \colon E \to E'$ is a 1-cell and $\lambda \colon h\hat{f} \to \hat{f'}h$ is a 2-cell in $\mathcal{K}$.
  \item[2-cell] $\alpha \colon h \to h'$ such that $\alpha \circ \lambda_1 = \lambda_2 \circ \alpha$.
  \end{description}
\end{definition}

When choosing $\cat{K} = \cat{Cat}$, we recover a 2-category of functor liftings.
Note the similarity between this 2-category and
the 2-category of monads \cite{street1972formal}. The next result has recently been
proved in a more general setting \cite[Theorem VII.6]{DBLP:conf/lics/AmorimKS25}, but for the sake of
concreteness, we present it in a more direct fashion---the proof can be found in the appendix.

\begin{theorem}
  \label{th:commagluing}
  For each 1-cell $(h, \lambda)\colon (E, p, \hat{f}) \to (E', p',
  \hat{f}')$, $\cat{Lift}(f)$ has the comma object of $(h, \lambda)$
  along $\id$, and it is computed pointwise.
\end{theorem}

Therefore, when instantiating this to our application, this implies that there
is a ``glued'' functor $\liftF \colon \cat{X} \to \cat{X}$ that universally projects
down to $\hat{F}$ and $\hat{F^g}$. We now show how this interacts with their respective
categories of coalgebras.
\begin{theorem}[cf. \Cref{ap:commacoalg}]
  \label{th:commacoalg}
  The category $\Coalg{\liftF}$ is the vertex of the comma object
  of the functor $\Coalg{\ladjDelta}_\alpha$ as defined
  in \Cref{def:coalg} along $\id$.
\end{theorem}

Combining this theorem with \Cref{lem:adjointcomma},
we can prove that there are adjunctions between the comma and base
categories.

\begin{lemma}
  The forgetful functor $ \cat{Coalg}(\liftF) \to \cat{Coalg}(\hat{F^g})$ has a
  right adjoint.
\end{lemma}

We prove our abstract soundness theorem by unfolding the construction above and
composing the functors
$\cat{Coalg}(\hat{F^g}) \to \cat{Coalg}(\liftF) \to \cat{Coalg}(\hat{F})$,

\begin{theorem}[Soundness] \label{cor:soundness}
If $\gradedRelDef{}$ is a graded simulation on a coalgebra $c$, then
$\ladjDelta[\gradedRelDef{g}]$ is an ungraded simulation on $c$.
\end{theorem}

In \Cref{th:commacoalg}, we see that the coalgebra construction preserves
comma objects. Another interesting result we can show is that under the
assumptions of \doctrine{}s, the glued functors perfectly characterize
the \doctrine{} structure.
See \cref{ap:proof_gluedequiv} for the proof.

\begin{theorem}
  \label{th:gluedequiv}
  Assuming that $\ladjDelta \colon \cat{E}_g \to \cat{E}$ is a left adjoint,
  every $F$-lifting in $\cat{X}$ restricts to a doctrine structure on $\cat{E}_g$
  and $\cat{E}$. Furthermore, this operation is a projection to the operation
  that glues together two functors.
\end{theorem}

\subsection{Final coalgebra}

Final coalgebras provide useful coinductive reasoning principles and
are widely used in coalgebraic automata theory. When it
comes to fibred bisimulation, the final coalgebra provides a useful
way of characterizing which invariants are kept track of by
bisimulation relations. For example, in the DFA example, the final
coalgebra has languages over the alphabet as its underlying set and
equality is the relational component, i.e. bisimulation is
equivalent to language equality.

In some cases, we can explicitly characterize the final coalgebra. A
useful result due to Ad\'{a}mek~\cite{DBLP:journals/tcs/Adamek03} allows us to compute
final coalgebra of functors satisfying certain continuity properties

\begin{theorem}
  Let $F \colon \cat{C} \to \cat{C}$ be an endofunctor on a category with
  terminal object such that the limit 
  of the $\omega^\op$-chain $(\cdots \xrightarrow{F^2!} F^21 \xrightarrow{F!} F1 \xrightarrow{!}1)$ 
  exists and is preserved
  by $F$. In this case the final $F$-coalgebra is given by the limiting
  object with the coalgebra structure defined using the universal
  property of limits.
\end{theorem}

We now show how to calculate the final coalgebra of the glued
model. We begin by noting that because we have the adjunction
$(\ladjDelta) \dashv \Delta$, the category $\cat{X}$ is isomorphic
to the comma category $(id \downarrow \Delta)$.
Furthermore, since $\Delta$ is a right adjoint, it preserves limits, so
we can apply the following theorem:
\begin{lemma}[{cf.~\cite[Proposition 2.16.1]{borceux1994handbookI}}]
  \label{lem:commacomplete}
  Consider two complete categories $\cat{A}$, $\cat{B}$, and two
  limit-preserving functors $F \colon \cat{A} \to \cat{C}$, $G \colon \cat{B} \to \cat{C}$.
  The comma category $F \downarrow G$ is also complete and the projection
  functors preserve limits.
\end{lemma}

The restriction of the result above to limits of arbitrary shapes also holds.
We now proceed to characterize the final coalgebra of $\mathcal{F}$.  
After presenting a lemma obtained as a direct consequence of the lemma above,
we establish the theorem providing the characterisation --- See \Cref{ap:proof_finalglued} for a proof.
\begin{lemma}
  If $F$ and $G$ preserve a limit of some shape, then so does
  their gluing.
\end{lemma}
\begin{theorem} \label{thm:finalglued}
  The carrier of the final coalgebra of $\liftF$ has components $\nu
  \hat{F^g}$ and $\nu \hat{F}$, where $\nu \hat{F^g}$ and $\nu
  \hat{F}$ are the carriers of final coalgebras for $\hat{F^g}$ and
  $\hat{F}$.
\end{theorem}

\section{Examples} \label{sec:examples}
In this section we present a few examples that illustrate the range of
\doctrine{} as an abstraction. We begin by presenting a first example
based on reward-sensitive notions of simulations on labelled transition systems with rewards. Then,
we present a reward-sensitive variant of simulations on Markov Decision
Processes (MDPs), and then we move on to a pseudometric and
relation-based bisimulations on labelled Markov processes.

\subsection{Amortised bisimulations on labelled transition systems}
\newcommand{\Flts}{F_\mathrm{LTS}}
\newcommand{\Fhatlts}{\hat{F}_\mathrm{LTS}}
Kiehn and Arun-Kumar~\cite{DBLP:conf/forte/KiehnA05} have introduced
\emph{amortised bisimulations} for labelled transition systems 
where certain actions have a reward
associated with it. 
Their definition is remarkably similar to our notion of graded
simulation for \dar. 
Here, we slightly generalize their setting
by allowing rewards to be defined locally for each transition,
rather than globally for each action as in their automata class,
which
also refines the
usual notion of LTS bisimulation.

Let $\mathcal{P}\colon \Set \to \Set$ be the covariant powerset functor.
Consider a labelled transition system with rewards (LTS)
$\delta\colon S \to \mathcal{P}(S \times \nat)^\Sigma$,
where $\Sigma$ is a set of labels and $(s', n) \in \delta(s)(a)$ means a possible transition 
$s \xrightarrow{a} s'$ with a reward $n$.
We write $\Flts$ for the functor $\mathcal{P}(- \times \nat)^\Sigma\colon \Set \to \Set$.

For liftings of $\clat$-fibrations, we only provide  
the action of functors on objects because any $\clat$-fibration is faithful.

\paragraph*{Ungraded bisimulations on LTSs}
Define a lifting $\Fhatlts \colon \cat{Rel} \to \cat{Rel}$ of $\Flts$ along $p\colon \cat{Rel} \to \Set$
  by 
  \begin{displaymath}
  \Fhatlts(R \subseteq X \times X) \coloneqq
  \left\{(\tau_0, \tau_1) \ \middle \vert \
  \begin{array}{l}
    \forall i \in \{0, 1\}.~\forall a \in \Sigma.~\forall (x_i, r_i) \in \tau_i(a). \\
    \quad \exists (x_{1-i}, r_{1-i}) \in \tau_{1-i}(a).~(x_0, x_1) \in R
  \end{array}
  \right\}.
  \end{displaymath}
A coalgebra $\delta \colon R \dotarrow \Fhatlts (R)$
is an ungraded bisimulation on the LTS $\delta$,
which is a standard bisimulation relation on the underlying system $(\mathcal{P}\pi)^\Sigma \circ \delta$ where $\pi\colon S \times \nat \to S$ is the projection.

\paragraph*{Graded bisimulations on LTSs}
Using the fibration $\nat$-graded $\fibgraded{p}{{\nat}}\colon \totalgraded{\cat{Rel}}{{\nat}} \to \Set$ introduced in \Cref{eg:dar},
we define a lifting $\Fhatlts^{\nat}\colon \totalgraded{\cat{Rel}}{{\nat}} \to \totalgraded{\cat{Rel}}{{\nat}}$ of $\Flts$. Its actions on objects $\Fhatlts^{\nat}(\gradedRelDef{})_n$ is
defined as.
  \begin{align*} 
  \left\{ (\tau_0, \tau_1) \ \middle \vert  \
  \begin{array}{l}
    \forall i \in \{0, 1\}.~\forall a \in \Sigma.~\forall (x_i, r_i) \in \tau_i(a).~\exists (x_{1-i}, r_{1-i}) \in \tau_{1-i}(a). \\
    \quad \big(r_0 \leq n+r_1   \text{ and }
    (x_0, x_1) \in \gradedRelDef{n-r_0 + r_1}\big)
  \end{array}
    \right\}.
  \end{align*}
A coalgebra $\delta\colon \gradedRelDef{} \dotarrow \Fhatlts^{\nat}(\gradedRelDef{})$ is a graded bisimulation on the LTS, known as a amortised bisimulation~\cite{DBLP:conf/forte/KiehnA05}.
It is a reward-sensitive refinement of bisimulations given as $\Fhatlts$-coalgebras.

\paragraph*{\Doctrine{}}
The structures introduced above, together with the appropriate fibred adjunction~\eqref{eq:fibadj}, form an \doctrine{}.
Note that 
there exists a canonical natural transformation $(\ladjDelta) \circ \Fhatlts^{\nat} \Rightarrow \Fhatlts \circ (\ladjDelta)$
because $\Fhatlts = (\ladjDelta) \circ \Fhatlts^{\nat} \circ \Delta$ holds.
Hence, \Cref{cor:soundness} yields the following result.
\begin{proposition}
  For an LTS $\delta\colon S \to \mathcal{P}(S \times \nat)^\Sigma$,
  if $\gradedRelDef{}$ is a graded bisimulation on the LTS then $\ladjDelta[\gradedRelDef{}]$ is an ungraded bisimulation on the LTS $\delta$.
\end{proposition}

\subsection{Probabilistic simulations and their reward-sensitive variant}

\newcommand{\Fmdp}{F_\mathrm{MDP}}
\newcommand{\Fhatmdp}{\hat{F}_\mathrm{MDP}}

Let 
$\subdist$ be the subdistribution monad over $\Set$.
Consider a Markov Decision Process
$\langle \delta, \rew \rangle \colon S \to (\subdist S \times {\R^+})^\Sigma$,
where 
$\R^+$ is the set of positive real numbers,
$\delta\colon S \to (\subdist S)^\Sigma$ describes probabilistic transitions and $\rew \colon S \to (\R^+)^\Sigma$ assigns an expected reward for each transition.
We write $\Fmdp$ for the functor $(\subdist (-) \times {\R^+})^\Sigma\colon \Set \to \Set$.

\paragraph*{Ungraded simulations on MDPs}
We first introduce simulations that ignore rewards,
capturing only the probabilistic behaviour of the underlying MDP $\delta$.
To capture probabilistic simulations,
we lift $\Fmdp$ along the fibration $p\colon \cat{Rel} \to \Set$.
For a relation $R \subseteq I \times I$ and a subset $X \subseteq I$,
We write $R(X)$ for the set $\{x' \mid \exists x \in X.~(x, x') \in R\}$,
and for $f \in \subdist(S)$ and $X \subseteq S$,
we write $f(X)$ for $\sum_{x \in X}f(x)$.
Define a lifting $\Fhatmdp \colon \cat{Rel} \to \cat{Rel}$ of $\Fmdp$ along $p$
  by 
  \begin{displaymath}
  \Fhatmdp(R \subseteq I \times I) = \{(f, g) \mid \forall a \in \Sigma.~\forall X \subseteq I.~(\pi_1 f(a))(X) \leq (\pi_1 g(a))(R(X))\}.
  \end{displaymath}
A coalgebra $\langle \delta, \rew \rangle\colon R \dotarrow \Fhatmdp (R)$
is an ungraded simulation on the MDP, which
coincides with $0$-simulations on labelled markov processes~\cite{DBLP:conf/qest/DesharnaisLT08}.

\paragraph*{Graded simulations on MDPs}
To account for rewards, we introduce a graded extension of the above lifting.
We do this by instantiating the fibration introduced in \Cref{eg:dar} with
$\R^+$ as our grading monoid, resulting in a fibration $\fibgraded{p}{{\R^+}}\colon \totalgraded{\cat{Rel}}{{\R^+}} \to \Set$.
We define a lifting $\Fhatmdp^{\R^+}\colon \totalgraded{\cat{Rel}}{{\R^+}} \to \totalgraded{\cat{Rel}}{{\R^+}}$
  of $\Fmdp$. Its action on objects $\Fhatmdp^{\R^+}(\gradedRelDef{})_r$ is.
  \begin{align*}
  \left\{ (\tau_0, \tau_1) \ \middle \vert  \
  \begin{array}{l}
    \forall a \in \Sigma.~\forall X \subseteq I.~\big(r_0 \leq r+ r_1   \text{ and }
    (\tau_0(X) \leq \tau_1(\gradedRelDef{r-r_0 + r_1} (X))\big) \\ 
    \quad \text{where }(\tau_i, r_i) = \tau_i(a) \text{ for each }i \in \{0, 1\}.
  \end{array}
    \right\}
  \end{align*}
  There the grade $r$ represents the admissible reward difference between corresponding transitions.
A coalgebra $\langle \delta, \rew \rangle \colon \gradedRelDef{} \dotarrow \Fhatmdp^{\R^+}(\gradedRelDef{})$ is a graded simulation on the MDP, which can be regarded as a reward-sensitive refinement of simulations given as $\Fhatmdp$-coalgebras.

\paragraph*{\Doctrine{}.}
The structures introduced above, together with the fibred adjunction~\eqref{eq:fibadj} for this setting, form an \doctrine{}.
Note that 
there exists a canonical natural transformation 
$\ladjDelta \circ \Fhatmdp^{\R^+} \Rightarrow \Fhatmdp \circ \ladjDelta$
because $\Fhatmdp = (\ladjDelta) \circ \Fhatmdp^{\R^+} \circ \Delta$ holds.
Hence, \Cref{cor:soundness} implies the following.
\begin{proposition}
  For an MDP $\langle \delta, \rew\rangle\colon S \to (\subdist S \times {\R^+})^\Sigma$,
  if $\gradedRelDef{}$ is a graded simulation on $\langle \delta, \rew\rangle$ then $\ladjDelta[\gradedRelDef{}]$ is a simulation on the underlying MDP $\delta$.
\end{proposition}

\subsection{Pseudometric and relational approximate notions of bisimulations}
\newcommand{\Flmp}{F_\mathrm{LMP}}
\newcommand{\Fhatlmp}{\hat{F}_\mathrm{LMP}}
\newcommand{\PMet}{\cat{PMet}_1}
\newcommand{\IRel}{\cat{IRel}}
\newcommand{\cball}{R}
\newcommand{\cballleft}{L}

We next turn to approximate notions of bisimulations for labelled Markov processes (LMPs), which can be represented as
coalgebras $\delta \colon X \to (\subdist X)^\Sigma$,
where $\subdist$ is the subdistribution monad and $\Sigma$ is a set of labels.
We write $\Flmp$ for the functor $(\subdist (-))^\Sigma\colon \Set \to \Set$.

In recent work, Desharnais and Sokolova \cite{DBLP:conf/csl/DesharnaisS26}
investigated how the \emph{L\'evy-Prokhorov pseudometrics}~\cite{Prokhorov1956ConvergenceOR} characterises $\epsilon$-bisimulations, relating the metric and relational viewpoints.
We now show how this is modeled by our framework.

Throughout this subsection,
we employ the unit interval $([0, 1], \oplus, 0)$,
where $\oplus$ is truncated addition, i.e., $a \oplus b = \min(1, a + b)$,
as a grading quantale.
We note, however, that our framework (excluding results~\Cref{thm:qrelmet} and \Cref{prop:bisimilarity_lmp}) can be applied to any quantale.

\paragraph*{(Pseudometric-based) ungraded bisimulations on LMPs}
We define $\PMet$ to be the category of $1$-bounded pseudometric spaces.
Its objects are pairs of a set $X$ and a function $d \colon X \times X \to [0, 1]$
such that $d$ satisfies $d(x, x) = 0$, $d(x, z) \leq d(x, y) + d(y, z)$, and $d(x, y) = d(y, x)$ for each $x, y, z \in X$.
A morphism $f\colon (X, d_X) \to (Y, d_Y)$ is a $1$-Lipschitz function, i.e.~a function $f \colon X \to Y$ such that $d_Y(f(x), f(x')) \leq d_X(x, x')$ for each $x, x' \in X$.
The forgetful functor $\PMet \to \Set$ is a $\clat$-fibration~\cite{DBLP:conf/fsttcs/BaldanBKK14}.

Define $\Fhatlmp\colon \PMet \to \PMet$ of $\Flmp$ along the fibration by
\begin{displaymath}
  \Fhatlmp(d)(f_0, f_1) \coloneqq \inf \{r \mid \forall a \in \Sigma.~\forall X \subseteq S.~f_i(a)(X) \leq f_{1-i}(a)(X_r) + r, i = 0, 1 \},
\end{displaymath}
where $X_r = \{y \mid \exists x \in X.~d(x, y) \leq r\}$.
It coincides with the L\'evy-Prokhorov lifting~\cite{Prokhorov1956ConvergenceOR,DBLP:conf/csl/DesharnaisS26},
except that we use closed balls instead of open ones.
This change allows the lifting $\Fhatlmp$ to be expressed as a composition $L \circ \Fhatlmp^{[0, 1]} \circ R$ , where $L \dashv R$ is an adjunction and $\Fhatlmp^{[0, 1]}$ is a lifting for graded relations defined below.

A coalgebra $\delta \colon (S, d) \dotarrow \Fhatlmp(S, d)$ is an ungraded bisimulation on the LMP.
\paragraph*{(Relation-based) graded bisimulations on LMPs}

In the same way as \Cref{eg:dar_mon} with letting $I \coloneqq ([0, 1], \oplus, 0)$,
we obtain the lax monoidal $[0, 1]$-graded fibrations $\ol{\totalgraded{\cat{Rel}}{[0, 1]}} \to \Set$ of $p\colon \cat{Rel} \to \Set$ (see \Cref{def:gradedfib}).

We write $\IRel$ for 
the total category $\ol{\totalgraded{\cat{Rel}}{[0, 1]}}$,
whose objects
are $([0, 1], \leq)$-graded binary relations $\gradedRelDef{}$ on a set $X$ such that 
  \begin{itemize}
  \item For every $x \in X$, $x \gradedRelDef{0} x$, and
  \item If $x_1 \gradedRelDef{r} x_2$ and $x_2 \gradedRelDef{r'} x_3$ then
    $x_1 \gradedRelDef{r \oplus r'} x_3$.
  \end{itemize}
A morphism between such relations is a function between their underlying
sets such that they laxly preserve the relation. 

We define a lifting
$\Fhatlmp^{[0, 1]}\colon \IRel \to \IRel$ of $\Flmp$ by specifying its action on
objects $\Fhatlmp^{[0, 1]}(\gradedRelDef{})_r$ as
\begin{displaymath}
  \{(f_0, f_1) \mid \forall a \in \Sigma.~\forall X \subseteq S.~f_i(a)(X) \leq f_{1-i}(a)(\gradedRelDef{r}(X)) + r, i = 0, 1 \}.
\end{displaymath}

A coalgebra $\delta \colon \gradedRelDef{} \dotarrow \Fhatlmp^{[0, 1]}(\gradedRelDef{})$ is a graded bisimulation on the LMP.
A component $\gradedRelDef{\epsilon}$ of a graded bisimulation $\gradedRelDef{}$ 
coincides with the so-called $\epsilon$-bisimulation introduced in~\cite{DBLP:conf/qest/DesharnaisLT08}.

\paragraph*{\Doctrine{}}
The following structures form an \doctrine.
\begin{equation*}
\begin{tikzcd}
	{\PMet} && {\IRel} \\
	& {\cat{Set}}
	\arrow["{\Fhatlmp}", from=1-1, to=1-1, loop, in=160, out=200, distance=5mm]
	\arrow[""{name=0, anchor=center, inner sep=0}, "{\cball}", curve={height=-6pt}, from=1-1, to=1-3]
	\arrow[from=1-1, to=2-2]
	\arrow[""{name=1, anchor=center, inner sep=0}, "{\cballleft}", curve={height=-6pt}, from=1-3, to=1-1]
	\arrow["{\Fhatlmp^{[0, 1]}}"', from=1-3, to=1-3, loop, in=20, out=340, distance=5mm]
	\arrow[from=1-3, to=2-2]
	\arrow["{\Flmp}"', from=2-2, to=2-2, loop, in=20, out=340, distance=5mm]
	\arrow["\dashv"{anchor=center, rotate=90}, draw=none, from=1, to=0]
\end{tikzcd}
\end{equation*}
where in the fibred adjunction $\cballleft \dashv \cball\colon \PMet \to \IRel$ the right adjoint uses the closed balls indexed by their radius
$(\cball(d))_r \coloneqq \{(x, x') \mid d(x, x') \leq r\}$
and the left adjoint is
$\cballleft(\gradedRelDef{})(x, x') \coloneqq \inf\{r \in [0, 1] \mid (x, x') \in \gradedRelDef{r}\}$. Note that in the definition of $L$ we must use $\inf$ due to the order on distances
being contravariant with respect to relational inclusion.
There exists a canonical natural transformation
$\ladjDelta \circ \Fhatlmp^{[0, 1]} \Rightarrow \Fhatlmp \circ \ladjDelta$
because $\Fhatlmp = \cballleft \circ \Fhatlmp^{[0, 1]} \circ \cball$ holds.
Note that the L\'evy-Prokhorov lifting (using open balls) also fits into our framework.
By \Cref{cor:soundness}, we obtain the following result.
\begin{theorem} \label{thm:qrelmet}
  If $\gradedRelDef{}\in \IRel$ is a bisimulation on the LMP $\delta$,
  then $\cballleft(\gradedRelDef{}) \in \PMet$ is also a bisimulation on $\delta$.
\end{theorem}

The converse direction, that a bisimulation in the left-hand side of the \doctrine{} diagram yields one in the right-hand side, does not necessarily hold in general, nor is it 
guaranteed by our abstract framework.
However, for finite LMPs, this follows by a direct argument.
See \Cref{ap:proof_qmetrel} for the proof.
\begin{theorem} \label{thm:qmetrel}
  Assume that the state space $S$ of the LMP $\delta$ is finite.
  If $d \in \PMet$ is a bisimulation on $\delta$,
  then $\cball(d) \in \IRel$ is also a bisimulation on $\delta$.
\end{theorem}
The authors of \cite{DBLP:conf/csl/DesharnaisS26} have proved that when $d$ is a fixed point of the corresponding operator, then the induced relation is a bisimulation.
Here, by adopting the closed-ball variant of the L\'evy-Prokhorov lifting and assuming that the state space is finite, the statement holds for bisimulations, not only for fixed points.

These correspondences between metric-based and relation-based bisimulations 
yield the following result. See \Cref{ap:proof_bisimilarity_lmp} for the proof.
\begin{proposition} \label{prop:bisimilarity_lmp}
  In the setting of \Cref{thm:qmetrel},
  $\cballleft$ maps the bisimilarity $\nu(\delta^*\Fhatlmp)$ to the bisimilarity $\nu(\delta^*\Fhatlmp^{[0, 1]})$.
\end{proposition}

\section{Conclusion}

In this paper we have introduced a framework for axiomatizing the interaction
between graded and ungraded notions of bisimulation. In order to be able to
treat such situations coalgebraically, we have proposed the application of
categorical gluing to coalgebraic automata theory. This technique
allows for the canonical construction of coalgebraic models that account for
the coherent integration of these different notions of bisimulation. We corroborates
the naturality of this model by using it to prove an abstract soundness theorem,
as well as showing that is had nice coalgebraic properties and
has interesting concrete examples.

\bibliographystyle{splncs04}
\bibliography{mybib}

@article{hermida1999some,
  title   = {Some properties of {Fib} as a fibred 2-category},
  author  = {Hermida, Claudio},
  journal = {Journal of Pure and Applied Algebra},
  year    = {1999}
}

@inproceedings{DBLP:conf/lics/AmorimKS25,
  author    = {Pedro H. Azevedo de Amorim and
               Satoshi Kura and
               Philip Saville},
  title     = {Logical relations for call-by-push-value models, via internal fibrations
               in a 2-category},
  booktitle = {{LICS}},
  pages     = {732--747},
  publisher = {{IEEE}},
  year      = {2025}
}

@book{borceux1994handbookI,
  title     = {Handbook of categorical algebra: Basic category theory},
  author    = {Borceux, Francis},
  volume    = {1},
  year      = {1994},
  publisher = {Cambridge University Press}
}

@article{DBLP:journals/iandc/HermidaJ98,
  author  = {Claudio Hermida and
             Bart Jacobs},
  title   = {Structural Induction and Coinduction in a Fibrational Setting},
  journal = {Inf. Comput.},
  volume  = {145},
  number  = {2},
  pages   = {107--152},
  year    = {1998}
}

@book{DBLP:books/daglib/0023251,
  author    = {Bart P. F. Jacobs},
  title     = {Categorical Logic and Type Theory},
  series    = {Studies in logic and the foundations of mathematics},
  volume    = {141},
  publisher = {North-Holland},
  year      = {2001},
  isbn      = {978-0-444-50853-9},
  bibsource = {dblp computer science bibliography, https://dblp.org}
}

@article{DBLP:journals/ngc/KomoridaKHKHEH22,
  author  = {Yuichi Komorida and
             Shin{-}ya Katsumata and
             Nick Hu and
             Bartek Klin and
             Samuel Humeau and
             Clovis Eberhart and
             Ichiro Hasuo},
  title   = {Codensity Games for Bisimilarity},
  journal = {New Gener. Comput.},
  volume  = {40},
  number  = {2},
  pages   = {403--465},
  year    = {2022}
}

@phdthesis{DBLP:books/daglib/0072949,
  author = {Claudio Hermida},
  title  = {Fibrations, logical predicates and indeterminates},
  school = {University of Edinburgh, {UK}},
  year   = {1993}
}

@article{street1972formal,
  title   = {The formal theory of monads},
  author  = {Street, Ross},
  journal = {Journal of Pure and Applied Algebra},
  year    = {1972}
}

@inproceedings{DBLP:conf/csl/DesharnaisS26,
  author    = {Jos{\'{e}}e Desharnais and
               Ana Sokolova},
  title     = {{\(\epsilon\)}-Distance via {L}{\'{e}}vy-{P}rokhorov Lifting},
  booktitle = {{CSL}},
  series    = {LIPIcs},
  volume    = {363},
  pages     = {26:1--26:24},
  publisher = {Schloss Dagstuhl - Leibniz-Zentrum f{\"{u}}r Informatik},
  year      = {2026}
}

@article{HERRLICH1974125,
  title    = {Topological functors},
  journal  = {General Topology and its Applications},
  volume   = {4},
  number   = {2},
  pages    = {125-142},
  year     = {1974},
  issn     = {0016-660X},
  author   = {Horst Herrlich}
}

@article{DBLP:journals/logcom/SprungerKDH21,
  author  = {David Sprunger and
             Shin{-}ya Katsumata and
             J{\'{e}}r{\'{e}}my Dubut and
             Ichiro Hasuo},
  title   = {Fibrational bisimulations and quantitative reasoning: Extended version},
  journal = {J. Log. Comput.},
  volume  = {31},
  number  = {6},
  pages   = {1526--1559},
  year    = {2021}
}

@article{DBLP:journals/jlp/KurzV16,
  author  = {Alexander Kurz and
             Jir{\'{\i}} Velebil},
  title   = {Relation lifting, a survey},
  journal = {J. Log. Algebraic Methods Program.},
  volume  = {85},
  number  = {4},
  pages   = {475--499},
  year    = {2016}
}

@article{DBLP:journals/iandc/Ibaraki76,
  author  = {Toshihide Ibaraki},
  title   = {Finite Automata Having Cost Functions},
  journal = {Inf. Control.},
  volume  = {31},
  number  = {2},
  pages   = {153--176},
  year    = {1976}
}

@incollection{zawadowski,
  author    = {Marek Zawadowski},
  title     = {Lax Monoidal Fibrations},
  booktitle = {Models, Logics, and Higher-Dimensional Categories},
  series    = {CRM Proceedings and Lecture Notes},
  volume    = {53},
  pages     = {341--426},
  publisher = {Amer. Math. Soc.},
  year      = {2011}
}

@inproceedings{DBLP:conf/qest/DesharnaisLT08,
  author    = {Jos{\'{e}}e Desharnais and
               Fran{\c{c}}ois Laviolette and
               Mathieu Tracol},
  title     = {Approximate Analysis of Probabilistic Processes: Logic, Simulation
               and Games},
  booktitle = {{QEST}},
  pages     = {264--273},
  publisher = {{IEEE} Computer Society},
  year      = {2008}
}

@inproceedings{DBLP:conf/stacs/BeoharG0MFSW24,
  author    = {Harsh Beohar and
               Sebastian Gurke and
               Barbara K{\"{o}}nig and
               Karla Messing and
               Jonas Forster and
               Lutz Schr{\"{o}}der and
               Paul Wild},
  title     = {Expressive Quantale-Valued Logics for Coalgebras: An Adjunction-Based
               Approach},
  booktitle = {{STACS}},
  series    = {LIPIcs},
  volume    = {289},
  pages     = {10:1--10:19},
  publisher = {Schloss Dagstuhl - Leibniz-Zentrum f{\"{u}}r Informatik},
  year      = {2024}
}

@inproceedings{DBLP:conf/forte/KiehnA05,
  author    = {Astrid Kiehn and
               S. Arun{-}Kumar},
  title     = {Amortised Bisimulations},
  booktitle = {{FORTE}},
  series    = {Lecture Notes in Computer Science},
  volume    = {3731},
  pages     = {320--334},
  publisher = {Springer},
  year      = {2005}
}

@inproceedings{DBLP:conf/fsttcs/BaldanBKK14,
  author    = {Paolo Baldan and
               Filippo Bonchi and
               Henning Kerstan and
               Barbara K{\"{o}}nig},
  title     = {Behavioral Metrics via Functor Lifting},
  booktitle = {{FSTTCS}},
  series    = {LIPIcs},
  volume    = {29},
  pages     = {403--415},
  publisher = {Schloss Dagstuhl - Leibniz-Zentrum f{\"{u}}r Informatik},
  year      = {2014}
}

@inproceedings{DBLP:conf/calco/MiliusPS15,
  author    = {Stefan Milius and
               Dirk Pattinson and
               Lutz Schr{\"{o}}der},
  title     = {Generic Trace Semantics and Graded Monads},
  booktitle = {{CALCO}},
  series    = {LIPIcs},
  volume    = {35},
  pages     = {253--269},
  publisher = {Schloss Dagstuhl - Leibniz-Zentrum f{\"{u}}r Informatik},
  year      = {2015}
}

@inproceedings{DBLP:conf/concur/Glabbeek90,
  author    = {Rob J. van Glabbeek},
  title     = {The Linear Time-Branching Time Spectrum (Extended Abstract)},
  booktitle = {{CONCUR}},
  series    = {Lecture Notes in Computer Science},
  volume    = {458},
  pages     = {278--297},
  publisher = {Springer},
  year      = {1990}
}

@inproceedings{DBLP:conf/concur/DorschMS19,
  author    = {Ulrich Dorsch and
               Stefan Milius and
               Lutz Schr{\"{o}}der},
  title     = {Graded Monads and Graded Logics for the Linear Time - Branching Time
               Spectrum},
  booktitle = {{CONCUR}},
  series    = {LIPIcs},
  volume    = {140},
  pages     = {36:1--36:16},
  publisher = {Schloss Dagstuhl - Leibniz-Zentrum f{\"{u}}r Informatik},
  year      = {2019}
}

@article{DBLP:journals/iandc/Ibaraki78b,
  author  = {Toshihide Ibaraki},
  title   = {Finite Automata Having Cost Functions: Nondeterministic Models},
  journal = {Inf. Control.},
  volume  = {37},
  number  = {1},
  pages   = {40--69},
  year    = {1978}
}

@inproceedings{DBLP:conf/lics/AlurDDRY13,
  author    = {Rajeev Alur and
               Loris D'Antoni and
               Jyotirmoy V. Deshmukh and
               Mukund Raghothaman and
               Yifei Yuan},
  title     = {Regular Functions and Cost Register Automata},
  booktitle = {{LICS}},
  pages     = {13--22},
  publisher = {{IEEE} Computer Society},
  year      = {2013}
}

@article{schutzenberger1961definition,
  title   = {On the definition of a family of automata},
  author  = {Sch{\"u}tzenberger, Marcel Paul},
  journal = {Information and Control},
  year    = {1961}
}

@article{Tarski1955ALF,
  title   = {A LATTICE-THEORETICAL FIXPOINT THEOREM AND ITS APPLICATIONS},
  author  = {Alfred Tarski},
  journal = {Pacific Journal of Mathematics},
  year    = {1955},
  volume  = {5},
  pages   = {285-309}
}

@article{Prokhorov1956ConvergenceOR,
  title   = {Convergence of Random Processes and Limit Theorems in Probability Theory},
  author  = {Yu. V. Prokhorov},
  journal = {Theory of Probability and Its Applications},
  year    = {1956},
  volume  = {1},
  pages   = {157-214}
}

@article{grodin2024amortized,
  title     = {Amortized Analysis via Coalgebra},
  author    = {Grodin, Harrison and Harper, Robert},
  journal   = {Electronic Notes in Theoretical Informatics and Computer Science},
  volume    = {4},
  year      = {2024},
  publisher = {Episciences. org}
}

@article{DBLP:journals/tcs/Weber94,
  author  = {Andreas Weber},
  title   = {Finite-Valued Distance Automata},
  journal = {Theor. Comput. Sci.},
  volume  = {134},
  number  = {1},
  pages   = {225--251},
  year    = {1994}
}

@article{weber2007yoneda,
  title   = {Yoneda structures from 2-toposes},
  author  = {Weber, Mark},
  journal = {Applied Categorical Structures},
  year    = {2007}
}

@article{DBLP:journals/tcs/Adamek03,
  author  = {Jir{\'{\i}} Ad{\'{a}}mek},
  title   = {On final coalgebras of continuous functors},
  journal = {Theor. Comput. Sci.},
  volume  = {294},
  number  = {1/2},
  pages   = {3--29},
  year    = {2003}
}

\newpage
\appendix

\section{Omitted proofs}

\subsection{Proof of \Cref{prop:soundness}} \label{ap:proof_soundness_sim}
\begin{proof}
  For functors $A, B$ and natural transformations $\alpha, \beta$ with appropriate types,
  $\Coalg{B}_\beta \circ \Coalg{A}_\alpha = \Coalg{BA}_{\beta_A \circ B\alpha}$ holds.
  Therefore $p \circ (\ladjDelta) = \fibgraded{p}{\nat}$
  and $(\Fda\pi_1)_{\fibgraded{p}{\nat}} = p(\Fhatda\epsilon \circ \alpha_{M})$ concludes the proof.
  \qed
\end{proof}
\subsection{Proof of \Cref{lem:adjointcomma}}
\label{app:adjointcomma}
\begin{proof}
  We define the left adjoint by using the universal property of
  comma categories:
\[\begin{tikzcd}
	&& x \\
	x & y && {(f \downarrow \id)} & y \\
	x & y && x & y
	\arrow["l", dashed, from=1-3, to=2-4]
	\arrow["f", curve={height=-12pt}, from=1-3, to=2-5]
	\arrow["id", curve={height=12pt}, from=1-3, to=3-4]
	\arrow["f", from=2-1, to=2-2]
	\arrow["id"', from=2-1, to=3-1]
	\arrow[""{name=0, anchor=center, inner sep=0}, "id", from=2-2, to=3-2]
	\arrow["q", from=2-4, to=2-5]
	\arrow[""{name=1, anchor=center, inner sep=0}, "p"', from=2-4, to=3-4]
	\arrow["id", from=2-5, to=3-5]
	\arrow["f"', from=3-1, to=3-2]
	\arrow["\varphi", between={0.3}{0.7}, Rightarrow, from=3-4, to=2-5]
	\arrow["f"', from=3-4, to=3-5]
	\arrow["{=}"{description}, draw=none, from=0, to=1]
\end{tikzcd}\]
The unit is the identity. The counit is given by the two-dimensional universal
property of comma objects, where we have to first define its $p$ and $q$, which
we define as the following string diagrams:
\[
\begin{tikzcd}
	{(f \downarrow \id)} & x & {(f \downarrow \id)} & {(f \downarrow \id)} & x & {(f\downarrow \id)} \\
	& x &&& y
	\arrow["p", from=1-1, to=1-2]
	\arrow["p"', from=1-1, to=2-2]
	\arrow["l", from=1-2, to=1-3]
	\arrow[equals, from=1-2, to=2-2]
	\arrow["p", from=1-3, to=2-2]
	\arrow["p", from=1-4, to=1-5]
	\arrow[""{name=0, anchor=center, inner sep=0}, "q"', from=1-4, to=2-5]
	\arrow["l", from=1-5, to=1-6]
	\arrow["f", from=1-5, to=2-5]
	\arrow["q", from=1-6, to=2-5]
	\arrow["\varphi", between={0}{0.8}, Rightarrow, from=1-5, to=0]
\end{tikzcd}
\]
In order to show that this indeed defines a 2-cell, we have to show that it satisfies
the exchange law, which follows by the universal construction of $l$:

\[\begin{tikzcd}
	{(f \downarrow \id)} & x & {(f \downarrow \id)} & {(f \downarrow \id)} & x & {(f \downarrow \id)} & x \\
	&& x \\
	y && y & y &&& y
	\arrow["p", from=1-1, to=1-2]
	\arrow[""{name=0, anchor=center, inner sep=0}, "p"', from=1-1, to=2-3]
	\arrow["q"', from=1-1, to=3-1]
	\arrow["l", from=1-2, to=1-3]
	\arrow[from=1-2, to=2-3]
	\arrow["p", from=1-3, to=2-3]
	\arrow["p", from=1-4, to=1-5]
	\arrow[""{name=1, anchor=center, inner sep=0}, "q"', from=1-4, to=3-4]
	\arrow["l", from=1-5, to=1-6]
	\arrow["\varphi", shift left=5, between={0.3}{0.7}, Rightarrow, from=1-5, to=3-4]
	\arrow["f"', from=1-5, to=3-7]
	\arrow["p", from=1-6, to=1-7]
	\arrow[""{name=2, anchor=center, inner sep=0}, "q"', from=1-6, to=3-7]
	\arrow["f", from=1-7, to=3-7]
	\arrow["f", from=2-3, to=3-3]
	\arrow["id"', from=3-1, to=3-3]
	\arrow["id"', from=3-4, to=3-7]
	\arrow["\varphi"', between={0.2}{1}, Rightarrow, from=0, to=3-1]
	\arrow["\varphi"{pos=0.6}, between={0.2}{0.8}, Rightarrow, from=1-7, to=2]
	\arrow["{=}"{description}, draw=none, from=2-3, to=1]
\end{tikzcd}\]

Triangle laws:
The first triangle law follows by the following pasting diagram
\[\begin{tikzcd}
	{(f \downarrow \id)} & x & x & {(f \downarrow \id)} & x & x \\
	& {(f \downarrow \id)}
	\arrow["p", from=1-1, to=1-2]
	\arrow[""{name=0, anchor=center, inner sep=0}, "id"', from=1-1, to=2-2]
	\arrow["id", from=1-2, to=1-3]
	\arrow[""{name=1, anchor=center, inner sep=0}, "l"', from=1-2, to=2-2]
	\arrow["{=}"{description}, shift right=5, draw=none, from=1-3, to=1-4]
	\arrow["p", from=1-4, to=1-5]
	\arrow["p"', curve={height=18pt}, from=1-4, to=1-6]
	\arrow["id", from=1-5, to=1-6]
	\arrow[""{name=2, anchor=center, inner sep=0}, "p"', from=2-2, to=1-3]
	\arrow["{=}", shift right=2, draw=none, from=1, to=2]
	\arrow[between={0}{0.8}, Rightarrow, from=1-2, to=0]
\end{tikzcd}\]

The second triangle law uses the 2-dimensional universal property of comma objects.
\[\begin{tikzcd}
	x & {(f \downarrow \id)} & {(f \downarrow \id)} & x & x & {(f \downarrow \id)} & {(f \downarrow \id)} \\
	& x &&&& x & {(f \downarrow \id)} & x
	\arrow["l", from=1-1, to=1-2]
	\arrow[""{name=0, anchor=center, inner sep=0}, "id"', from=1-1, to=2-2]
	\arrow[""{name=1, anchor=center, inner sep=0}, "id", from=1-2, to=1-3]
	\arrow[""{name=2, anchor=center, inner sep=0}, "p"', from=1-2, to=2-2]
	\arrow["p", from=1-3, to=1-4]
	\arrow["{=}"{description}, draw=none, from=1-4, to=1-5]
	\arrow["l", from=1-5, to=1-6]
	\arrow["id"', from=1-5, to=2-6]
	\arrow["id", from=1-6, to=1-7]
	\arrow["p"', from=1-6, to=2-6]
	\arrow["p", from=1-7, to=2-8]
	\arrow["l"', from=2-2, to=1-3]
	\arrow["l"', from=2-6, to=2-7]
	\arrow[curve={height=-12pt}, from=2-6, to=2-8]
	\arrow["p"', from=2-7, to=2-8]
	\arrow["{=}"', draw=none, from=2, to=0]
	\arrow["\varepsilon", between={0.4}{0.8}, Rightarrow, from=2-2, to=1]
\end{tikzcd}\]

\[\begin{tikzcd}
	x & {(f \downarrow \id)} & {(f \downarrow \id)} & y & x & {(f \downarrow \id)} & {(f \downarrow \id)} \\
	& x &&&& x & {(f \downarrow \id)} & y \\
	{~} & x & {(f \downarrow \id)} & y \\
	&& x & {(f \downarrow \id)} & y
	\arrow["l", from=1-1, to=1-2]
	\arrow[""{name=0, anchor=center, inner sep=0}, "id"', from=1-1, to=2-2]
	\arrow[""{name=1, anchor=center, inner sep=0}, "id", from=1-2, to=1-3]
	\arrow[""{name=2, anchor=center, inner sep=0}, "p"', from=1-2, to=2-2]
	\arrow["q", from=1-3, to=1-4]
	\arrow["{=}"{description}, draw=none, from=1-4, to=1-5]
	\arrow["l", from=1-5, to=1-6]
	\arrow["id"', from=1-5, to=2-6]
	\arrow["id", from=1-6, to=1-7]
	\arrow[""{name=3, anchor=center, inner sep=0}, "p"', from=1-6, to=2-6]
	\arrow["q", from=1-7, to=2-8]
	\arrow["l"', from=2-2, to=1-3]
	\arrow["l"', from=2-6, to=2-7]
	\arrow["f", curve={height=-12pt}, from=2-6, to=2-8]
	\arrow["q"', from=2-7, to=2-8]
	\arrow["{=}"{description}, draw=none, from=3-1, to=3-2]
	\arrow["l", from=3-2, to=3-3]
	\arrow["id"', from=3-2, to=4-3]
	\arrow["q", from=3-3, to=3-4]
	\arrow["id", from=3-4, to=4-5]
	\arrow["l"', from=4-3, to=4-4]
	\arrow["f", curve={height=-12pt}, from=4-3, to=4-5]
	\arrow["q"', from=4-4, to=4-5]
	\arrow["{=}"', draw=none, from=2, to=0]
	\arrow["\varphi"{description}, between={0.2}{0.8}, Rightarrow, from=3, to=1-7]
	\arrow["\varepsilon", between={0.4}{0.8}, Rightarrow, from=2-2, to=1]
\end{tikzcd}\]
\end{proof}

\subsection{Proof of \Cref{lem:comma_fib}} \label{ap:proof_comma_fib}
\begin{proof}
  For the sake of simplicity, we provide a simpler proof. $\clat$-fibrations
  over a category $\cat{C}$ can be alternatively defined, by the Grothendieck
  construction 2-isomorphism, as functors $\cat{C}^{op} \to \clat$. Therefore,
  by the fact that comma objects are defined pointwise in 2-categories of
  (2-)functors, it suffices to show that $\clat$ has comma objects.
  
  Consider the complete lattices $L_1$, $L_1'$ and $L_2$, and infima preserving
  monotonic functions $f \colon L_1 \to L_2$ and $f' \colon L_1' \to L_2$. Their comma
  complete lattice has $X = \{(x, x') \in L_1 \times L_1'\, \mid\, f(x) \subseteq f'(x') \}$
  as underlying set. We equip it with the componentwise ordering.

  The complete lattice structure of $X$ is defined using the property
  that, by the adjoint functor theorem, any poset with arbitrary infima
  has arbitrary suprema, and vice-versa. Therefore, is suffices to show
  that $X$ has arbitrary infima. By assumption, $L_1$ and $L_1'$ have
  infima, so we define
  \[\bigsqcap_\alpha (x_\alpha, x_\alpha') = \left ( \bigsqcap_\alpha x_\alpha, \bigsqcap_\alpha x_\alpha' \right )\]
  We can show that this point is indeed in $X$. This follows by the assumptions
  that $f$ and $f'$ preserve arbitrary infima, and that for every $\alpha$,
  $f(x_\alpha) \sqsubseteq f'(x_\alpha')$, so $\bigsqcap_\alpha f(x_\alpha) \sqsubseteq \bigsqcap_\alpha f'(x_\alpha')$.

  The comma object universal properties follows from a direct calculation and
  from the order structure of $X$ being componentwise.
\end{proof}
\subsection{Proof of \Cref{th:commagluing}}
\label{app:comma}
\begin{proof}
  Being computed pointwise means that the vertex is $((h \downarrow \id), t)$,
  where $t$ is defined by the universal property of the comma category $(h \downarrow \id)$:
\[\begin{tikzcd}
	{(h \downarrow \id)} & {E'} && {(h \downarrow \id)} & {E'} \\
	E & {E'} & {E'} & E & {(h\downarrow \id)} & {E'} \\
	& E & {E'} && E & {E'}
	\arrow["q", from=1-1, to=1-2]
	\arrow["p", from=1-1, to=2-1]
	\arrow["id", from=1-2, to=2-2]
	\arrow["{f'}", from=1-2, to=2-3]
	\arrow["q", from=1-4, to=1-5]
	\arrow["p", from=1-4, to=2-4]
	\arrow["t"{description}, dotted, from=1-4, to=2-5]
	\arrow["{f'}", from=1-5, to=2-6]
	\arrow[shorten >=8pt, shorten <= 8pt, Rightarrow, from=2-1, to=1-2]
	\arrow["h"{description}, from=2-1, to=2-2]
	\arrow[""{name=0, anchor=center, inner sep=0}, "f"', from=2-1, to=3-2]
	\arrow[""{name=1, anchor=center, inner sep=0}, "{f'}", from=2-2, to=3-3]
	\arrow["{=}"{description}, draw=none, from=2-3, to=2-4]
	\arrow["id", from=2-3, to=3-3]
	\arrow["f"', from=2-4, to=3-5]
	\arrow["q", from=2-5, to=2-6]
	\arrow["p", from=2-5, to=3-5]
	\arrow[from=2-6, to=3-6]
	\arrow["h", from=3-2, to=3-3]
	\arrow[shorten >=8pt, shorten <= 8pt,Rightarrow, from=3-5, to=2-6]
	\arrow["h", from=3-5, to=3-6]
	\arrow[shorten >=7pt, shorten <= 7pt,Rightarrow, from=0, to=1]
\end{tikzcd}\]
In order to show that this is indeed a comma object in $\cat{Lift}(f)$
we have to show
\begin{itemize}
  \item That the 2-cell above in $\cat{Cat}$ is indeed a 2-cell in $\cat{Lift}(f)$
  \item That it satisfies the one-dimensional universal property of comma objects
\end{itemize}

Such objects also have a two-dimensional property, but for the applications we are
interested in, the 1-dimensional universal property suffices.

\begin{description}
\item[2-cell:] Assume that we have 0-cells $(E, p_1, \hat{f})$ and
  $(E', p_2, \hat{f'})$. This case boils down to proving
  the following pasting diagram equality:
\[\begin{tikzcd}
	X && {E'} & X && {E'} \\
	& E \\
	X & E & {E'} & X & E & {E'}
	\arrow[""{name=0, anchor=center, inner sep=0}, "q", from=1-1, to=1-3]
	\arrow["p"', from=1-1, to=2-2]
	\arrow["t"', from=1-1, to=3-1]
	\arrow[""{name=1, anchor=center, inner sep=0}, "{f'}", from=1-3, to=3-3]
	\arrow["q", from=1-4, to=1-6]
	\arrow[""{name=2, anchor=center, inner sep=0}, "t"', from=1-4, to=3-4]
	\arrow["{f'}", from=1-6, to=3-6]
	\arrow["h"', from=2-2, to=1-3]
	\arrow["f", from=2-2, to=3-2]
	\arrow["p", from=3-1, to=3-2]
	\arrow["h", from=3-2, to=3-3]
	\arrow["p"', from=3-4, to=3-5]
	\arrow[""{name=3, anchor=center, inner sep=0}, "q"{description}, curve={height=-30pt}, from=3-4, to=3-6]
	\arrow["h"', from=3-5, to=3-6]
	\arrow["{=}"{description}, draw=none, from=1, to=2]
	\arrow["\psi", between={0.3}{0.7}, Rightarrow, from=2-2, to=0]
	\arrow["\lambda"'{pos=0.4}, between={0.2}{0.7}, Rightarrow, from=2-2, to=1]
	\arrow["\psi", between={0}{0.8}, Rightarrow, from=3-5, to=3]
\end{tikzcd}\]
which follows by the definition of $t$ above---note that the lower leg on the
diagram on the right hand side can be directly rewritten from the definition of $t$.
This shows that $\psi$ is indeed a 2-cell.
\item[One-dimensional:]
  Assume that there is a 2-cell
\[\begin{tikzcd}
	{(Y, q, \hat{g})} & {(E', p_2, \hat{f'})} \\
	{(E, p_1, \hat{f})} & {(E', p_2, \hat{f'})}
	\arrow["{(n, \beta)}", from=1-1, to=1-2]
	\arrow["{ (m, \alpha)}"', from=1-1, to=2-1]
	\arrow["id", from=1-2, to=2-2]
	\arrow["\varphi", shorten >= 8pt, shorten <= 8pt,Rightarrow, from=2-1, to=1-2]
	\arrow["{(h,\lambda)}"', from=2-1, to=2-2]
\end{tikzcd}\]
We want to construct a universal 1-cell $(Y, q, \hat{g}) \to (X, p_1\circ p, t)$. Remember that
1-cells in $\cat{Lift}(f)$ are pairs $(s, \sigma)$. It is possible to define
$s$ by once again using the universal property of comma objects for $\varphi$.

It remains to define a transformation $sg \to ts$. We will define this by using
the two-dimensional universal property of comma objects. After a long sequence
of calculations, we can define such a transformation. The uniqueness of the
1-cell $(Y, q, \hat{g}) \to (X, p_1\circ p, t)$ can be lifted from the universal property of
comma objects.
\end{description}
\end{proof}

\subsection{Proof of \Cref{th:gluedequiv}} \label{ap:proof_gluedequiv}

\begin{proof}
   We define this map by using a variant of \Cref{lem:adjointcomma},
   which will give us a right adjoint $\iota \colon \cat{E} \to \cat{X}$
   and a left adjoint $\iota_g \colon \cat{E}_g \to \cat{X}$.

  Using these adjunctions we can define functors $\hat{F} \colon \cat{E} \to
  \cat{E}$ and $\hat{F^g}\colon \cat{E}_g \to \cat{E}_g$ as $q \circ
  \mathcal{H} \circ \iota$ and $p \circ \mathcal{H} \circ \iota_g$,
  respectively.  The distributive law $(\ladjDelta) \circ \hat{F^g}
  \Rightarrow \hat{F}\circ (\ladjDelta)$ is given by the horizontal
  composition $\bar{\eta} * \id * \varphi$. Where $\bar{\eta}$ is
  pointwise the $\cat{X}$ morphism given by the unit of the $(\ladjDelta) \dashv \Delta$ adjunction and the identity.

  In order to show the projection property, we use the fact that  
  $p \circ \mathcal{H} \circ \iota_g$ computes the first component
  of the glued functor, which in this case it is defined to be
  $F_g$. A similar argument holds for the functor
  $q \circ \mathcal{H} \circ \iota$. Finally, their mediating morphism
  is calculated using the distributive law, which is recovered
  by the horizontal composition and using the triangle law of adjunctions.
\qed
\end{proof}

\subsection{Proof of \Cref{thm:finalglued}} \label{ap:proof_finalglued}
\begin{proof}
  By Ad\'{a}mek's result, the carrier of the final coalgebra is
  given by the limit of the ascending chain. We want to calculate
  $p(\nu \liftF)$ and $q(\nu \liftF)$. Because both calculations
  are basically the same, we focus on $p(\nu \liftF)$. The functor
  $p$ preserves limits, so we have the following limiting cone:
\[\begin{tikzcd}
	&& {p(\nu \liftF)} \\
	\cdots & {p(\liftF^31)} & {p(\liftF^21)} & {p(\liftF1)} & {p(1) \cong 1}
	\arrow[from=1-3, to=2-1]
	\arrow[from=1-3, to=2-5]
	\arrow["{p(\liftF^3!)}"', from=2-1, to=2-2]
	\arrow["{p(\liftF^2!)}"', from=2-2, to=2-3]
	\arrow["{p(\liftF!)}"', from=2-3, to=2-4]
	\arrow["{!}"', from=2-4, to=2-5]
\end{tikzcd}\]
By construction, $p\circ \liftF = \hat{F^g} \circ p$, which together
with $p(1) \cong 1$ implies that the diagram above is equal to
diagram whose limit is $\nu \hat{F^g}$, concluding the proof.
\end{proof}

\subsection{Proof of \Cref{th:commacoalg}}
\label{ap:commacoalg}
\begin{proof}
  We know that this comma object is given by the comma category of the
  underlying categories of coalgebras. The proof follows by unfolding
  the definitions of the categories and comparing their objects and
  morphisms.
  \begin{description}
  \item[Objects]
    An object of $\Coalg{\liftF}$ is a morphism $(X_1, X_2, h \colon \ladjDelta X_1 \to X_2) \to (\hat{F^g} X_1, \hat{F} X_2, \hat{F} h \circ \alpha)$ in the comma category $\ladjDelta \downarrow \id$. Therefore, this is given by a pair $(f_1 \colon X_1 \to \hat{F^g} X_1, f_2 \colon X_2 \to \hat{F} X_2)$ of morphisms such that
    \[\begin{tikzcd}
	    {\ladjDelta X_1} & {\ladjDelta \hat{F^g}X_1} \\
	    & {\hat{F}\ladjDelta X_1} \\
	    {X_2} & {\hat{F}X_2}
	    \arrow["{\ladjDelta f_1}", from=1-1, to=1-2]
	    \arrow["h"', from=1-1, to=3-1]
	    \arrow["\alpha", from=1-2, to=2-2]
	    \arrow["{\hat{F}h}", from=2-2, to=3-2]
	    \arrow["{f_2}"', from=3-1, to=3-2]
\end{tikzcd}\]
    
    When we unfold the definition of the objects in the comma category
    $\Coalg{\ladjDelta}_\alpha \downarrow \id$, we get triples $(f_1, f_2, h)$,
    where $f_1 \colon X_1 \to \hat{F^g} X_1$, $f_2 \colon X_2 \to \hat{F} X_2$ and
    $h \colon \Coalg{\ladjDelta}_\alpha(f_1) \to f_2$ is a coalgebra morphism.
    Unfolding this, we get that they must make the following diagram
    commute:
\[\begin{tikzcd}
	{\ladjDelta X_1} & {X_2} \\
	{\ladjDelta \hat{F^g}X_1} \\
	{\hat{F}\ladjDelta X_1} & {\hat{F}X_2}
	\arrow["h", from=1-1, to=1-2]
	\arrow["{\ladjDelta f_1}"', from=1-1, to=2-1]
	\arrow["{f_2}", from=1-2, to=3-2]
	\arrow["\alpha"', from=2-1, to=3-1]
	\arrow["{\hat{F}h}"', from=3-1, to=3-2]
\end{tikzcd}\]
Up to a permutation, these two diagrams are the same, making the objects equal.
\item[Morphisms]
  We will use $\mathbb{X}$ to denote objects in
  the comma category $\ladjDelta \downarrow \id$. In this case, a coalgebra morphism
  $(\mathbb{X} \to \liftF \mathbb{X}) \to (\mathbb{X}' \to \liftF \mathbb{X}')$
  consists of a pair morphisms $g_1$ and $g_2$ such that
\[\begin{tikzcd}
	{X_1} & {X_1'} & {X_2} & {X_2'} \\
	{\ladjDelta X_1} & {\ladjDelta X_1'} & {\ladjDelta X_2} & {\ladjDelta X_2'}
	\arrow["{g_1}", from=1-1, to=1-2]
	\arrow["{f_1}"', from=1-1, to=2-1]
	\arrow["{f_1'}", from=1-2, to=2-2]
	\arrow["{g_2}", from=1-3, to=1-4]
	\arrow["{f_2}"', from=1-3, to=2-3]
	\arrow["{f_2'}", from=1-4, to=2-4]
	\arrow["{\ladjDelta g_1}"', from=2-1, to=2-2]
	\arrow["{\ladjDelta g_2}"', from=2-3, to=2-4]
\end{tikzcd}\]
where $(f_1, f_2) \colon \mathbb{X} \to \liftF\mathbb{X}$ and
$(f'_1, f'_2) \colon \mathbb{X}' \to \liftF\mathbb{X}'$ are the
components of the domain and codomain coalgebras, respectively.
The fact that these must be morphisms in $\ladjDelta \downarrow \id$
means that the following diagrams must commute:
\[\begin{tikzcd}
	&& {\ladjDelta X_1} & {\ladjDelta \hat{F^g} X_1} \\
	{\ladjDelta X_1} & {\ladjDelta X'_1} & {X_2} & {\hat{F}X_2} \\
	{X_2} & {X_2'} & {\ladjDelta X_1'} & {\ladjDelta \hat{F^g} X_1'} \\
	&& {X_2'} & {\hat{F}X_2'}
	\arrow["{\ladjDelta f_1}", from=1-3, to=1-4]
	\arrow["h"', from=1-3, to=2-3]
	\arrow["{\hat{F}h\circ \alpha}", from=1-4, to=2-4]
	\arrow["{\ladjDelta g_1}", from=2-1, to=2-2]
	\arrow["h"', from=2-1, to=3-1]
	\arrow["{h'}", from=2-2, to=3-2]
	\arrow["{f_2}"', from=2-3, to=2-4]
	\arrow["{g_2}"', from=3-1, to=3-2]
	\arrow["{\ladjDelta f_1'}", from=3-3, to=3-4]
	\arrow["{h'}"', from=3-3, to=4-3]
	\arrow["{\hat{F}h'\circ \alpha}", from=3-4, to=4-4]
	\arrow["{f_2'}"', from=4-3, to=4-4]
\end{tikzcd}\]

Now, we unfold the definitions of morphisms for
$\Coalg{\ladjDelta}_\alpha \downarrow \id$. Assuming that
we have coalgebras following the same convention as above,
an arrow $(f_1, f_2, h \colon \Coalg{\ladjDelta}_\alpha(f_1) \to f_2) \to (f_1', f_2', h' :\Coalg{\ladjDelta}_\alpha(f_1') \to f_2')$ is given by a pair of coalgebra
morphisms $g_1 \colon f_1 \to f_1'$ and $g_2 \colon f_2 \to f_2'$ such that
\[\begin{tikzcd}
	{\ladjDelta X_1} & {\ladjDelta X'_1} \\
	{X_2} & {X_2'}
	\arrow["{\ladjDelta g_1}", from=1-1, to=1-2]
	\arrow["h"', from=1-1, to=2-1]
	\arrow["{h'}", from=1-2, to=2-2]
	\arrow["{g_2}"', from=2-1, to=2-2]
\end{tikzcd}\]
As we have shown in the ``objects'' part of this proof, the last two diagrams
commute that assumption that $h$ and $h'$ are coalgebra morphisms. This shows
that both the object and morphism parts of the categories are the same and
concludes the proof.
  \end{description}
\end{proof}

\section{Omitted proofs for \Cref{sec:examples}}
\subsection{Proof of \Cref{thm:qmetrel}} \label{ap:proof_qmetrel}
Note that $LR = \id$.

\begin{lemma} \label{lem:fhatrel_rl}
  Let $\gradedRelDef{} \in \cat{QRel}$ be an object above a finite set $S$.
  If $\gradedRelDef{} = RL \gradedRelDef{}$, then $\Fhatlmp^{[0, 1]} \gradedRelDef{} = RL \Fhatlmp^{[0, 1]} \gradedRelDef{}$.
\end{lemma}
\begin{proof}
  Note that $\gradedRelDef{} = RL \gradedRelDef{}$ means:
  for each $x, y \in S$ and $q \in Q$, 
  \begin{center}
  $\inf \{q' \mid (x, y) \in \gradedRelDef{q'}\} \leq q$ implies $(x, y) \in \gradedRelDef{q}$.
  \end{center}
  Let us fix $f_0, f_1 \in \Flmp S$ and $q \in Q$.
  Define $Q' \coloneqq \{q' \in Q \mid (f_0, f_1) \in (\Fhatlmp^{[0, 1]} \gradedRelDef{})_{q'}\}$.
  Note that $Q' \neq \emptyset$ since $1$ is always in $Q'$.
  Assume that $\inf Q' \leq q$, and we aim to show $(f_0, f_1) \in (\Fhatlmp^{[0, 1]} \gradedRelDef{})_q$.

  \begin{itemize}
    \item Case $\inf Q' < q$: In this case, there exists $q' \in Q'$ such that $q' < q$.
      It implies that $(f_0, f_1) \in (\Fhatlmp^{[0, 1]} \gradedRelDef{})_{q'} \subseteq (\Fhatlmp^{[0, 1]} \gradedRelDef{})_{q}$.
    \item 
  Case $\inf Q' = q$: In this case, $q \leq q'$ holds for each $q' \in Q'$.
  By the definition of $Q'$ and $\Fhatlmp^{[0, 1]}$,
  for each $a \in \Sigma$, $X \subseteq S$, and $i \in \{0, 1\}$,
  \begin{displaymath}
  f_i(a)(X) \leq \inf_{q' \in Q'} (f_{1-i}(a)(\gradedRelDef{q'}(X)) + q') = \big(\inf_{q' \in Q'} f_{1-i}(a)(\gradedRelDef{q'}(X))\big) + q.
  \end{displaymath}
  Thus $(f_0, f_1) \in (\Fhatlmp^{[0, 1]} \gradedRelDef{})_q$ will follow once we show 
  \begin{equation} \label{eq:inf_q}
  \inf_{q' \in Q'} f_{1-i}(a)(\gradedRelDef{q'}(X)) \leq f_{1-i}(a)(\gradedRelDef{q}(X))
  \end{equation}
  for each $a, X, i$.
  Suppose, towards a contradiction, that \eqref{eq:inf_q} fails for some $a, X, i$.
  Then for each $q' \in Q'$, we have $f_{1-i}(\gradedRelDef{q'}(X)) > f_{1-i}(\gradedRelDef{q}(X))$,
  whence $\gradedRelDef{q'}(X) \supsetneq \gradedRelDef{q}(X)$.
  If there exists $x \in X$ and $y \in S$ such that $(x, y) \in \gradedRelDef{q'} \setminus \gradedRelDef{q}$ for each $q' \in Q'$, then $\inf \{q'' \mid (x, y) \in \gradedRelDef{q''}\} \leq \inf Q' = q$.
  By the assumption $\gradedRelDef{} = RL \gradedRelDef{}$,
  it implies that $(x, y) \in \gradedRelDef{q}$, a contradiction.
  Hence, we have ($\star$): for each $x \in X$ and $y \in S$, there exists $q' \in Q'$ such that $(x, y) \not \in \gradedRelDef{q'} \setminus \gradedRelDef{q}$.

  Now pick any $q_0 \in Q'$ (possible since $Q' \neq \emptyset$).
  From $\gradedRelDef{q}(X) \subsetneq \gradedRelDef{q_0}(X)$, 
  there is $y_0 \in \gradedRelDef{q_0}(X) \setminus \gradedRelDef{q}(X)$,
  which also implies the existence of $x_0 \in X$ such that $(x_0, y_0) \in \gradedRelDef{q_0} \setminus \gradedRelDef{q}$.
  By $(\star)$, there exists $q_1 \in Q'$ such that $(x_0, y_0) \not \in \gradedRelDef{q_1} \setminus \gradedRelDef{q}$,
  i.e.~$(x_0, y_0) \not \in \gradedRelDef{q_1}$ since $(x_0, y_0) \not \in \gradedRelDef{q}$.
  It follows that $q_1 < q_0$ (otherwise $(x_0, y_0) \not \in \gradedRelDef{q_0}$),
  and
  $\gradedRelDef{q_{1}} \subsetneq \gradedRelDef{q_0}$ with
   $(x_0, y_0) \in \gradedRelDef{q_0} \setminus \gradedRelDef{q_1}$.
  By repeating the argument, we obtain 
  a strictly descending chain $\gradedRelDef{q_0} \supsetneq \gradedRelDef{q_1} \supsetneq \cdots$
  together with pairs $(x_i, y_i) \in \gradedRelDef{q_i} \setminus \gradedRelDef{q_{i+1}}$.
  It contradicts to the finiteness of $S$.
  \end{itemize}
  In all cases we conclude $(f_0, f_1) \in (\Fhatlmp^{[0, 1]} \gradedRelDef{})_q$.
  Hence $\Fhatlmp^{[0, 1]} \gradedRelDef{} = RL \Fhatlmp^{[0, 1]} \gradedRelDef{}$.
  \qed
\end{proof}

\begin{proof}[\Cref{thm:qmetrel}]
  Let $c\colon d \to \Fhatlmp d$ be a coalgebra above $\delta$.
  Because $Rd = RLRd$ holds by $LR=\id$
  and $Rd$ is still above the finite set $S$,
  \Cref{lem:fhatrel_rl} yields $\Fhatlmp^{[0, 1]} Rd = RL \Fhatlmp^{[0, 1]} Rd = R\Fhatlmp d$.
  Following the same construction as in \Cref{def:coalg},
  we obtain
  the coalgebra $Rd \xrightarrow{Rc} R \Fhatlmp d = \Fhatlmp^{[0, 1]} R d$ above $\delta$.
  \qed
\end{proof}

\subsection{Proof of \Cref{prop:bisimilarity_lmp}} \label{ap:proof_bisimilarity_lmp}
\begin{proof}
  For each bisimulation $d \in \PMet$,
  we have $d = \cballleft \cball (d) \sqsubseteq \cballleft (\nu(\delta^*\Fhatlmp))$
  since $\cballleft \cball  = \id$ and $\cball (d) \sqsubseteq \nu(\delta^*\Fhatlmp)$ by \Cref{thm:qmetrel}.
  By \Cref{thm:qrelmet}, $\cballleft (\nu(\delta^*\Fhatlmp))$ is itself a bisimulation.
  Hence $\cballleft (\nu(\delta^*\Fhatlmp))$ is the greatest bisimulation.
  \qed
\end{proof}

\end{document}